\documentclass{article}
\usepackage[british]{babel}
\usepackage{amsmath,amssymb,latexsym,theorem}
\usepackage[top=25.5mm, bottom=25.5mm, left=25.5mm, right=25.5mm]{geometry}
\usepackage{authblk}
\usepackage{url}
\usepackage{hyperref}
\usepackage{multirow}
\usepackage{algorithm}
\usepackage{algorithmic}
\usepackage{amsfonts}
\usepackage{comment}
\newcommand{\mb}{\mathbf}
\newcommand{\mc}{\mathcal}
\newcommand{\gf}{\mathrm{GF}}
\newcommand{\transpose}{\ensuremath{\mathsf{T}}}

\DeclareMathOperator{\OPT}{OPT}

\newcommand{\eps}{\ensuremath{\varepsilon}}

\newcommand{\abs}[1]{\ensuremath{\mathopen\lvert #1 \mathclose\rvert}}
\newcommand{\norm}[1]{\ensuremath{\mathopen\lVert #1 \mathclose\rVert}}
\newcommand{\Abs}[1]{\ensuremath{\left| #1 \right|}}

\newcommand{\tensor}{\ensuremath{\otimes}}

\newcommand{\zo}{\ensuremath{\{0,1\}}}
\newcommand{\pmo}{\ensuremath{\{-1,1\}}}
\newcommand{\zpmo}{\ensuremath{\{-1,0,1\}}}

\newcommand{\NP}{\ensuremath{\mathsf{NP}}}

\newcommand{\symdiff}{\ensuremath{\bigtriangleup}}

\newenvironment{proof}{\noindent\textbf{Proof\ }}{\hspace*{\fill}$\Box$\medskip}
\newtheorem{theorem}{Theorem}
\newtheorem{corollary}[theorem]{Corollary}

\newtheorem{definition}[theorem]{Definition}
\newtheorem{lemma}[theorem]{Lemma}
\newtheorem{remark}[theorem]{Remark}
{\theorembodyfont{\rmfamily\small}}

\begin{document}

\title{Low Rank Approximation of Binary Matrices: \\
Column Subset Selection and Generalizations}
\author[1]{Chen Dan\thanks{Email: chen\_dan@pku.edu.cn}}
\author[2]{Kristoffer Arnsfelt Hansen\thanks{Email: arnsfelt@cs.au.dk}}
\author[1]{He Jiang\thanks{Email: hejiang@pku.edu.cn}}
\author[1]{Liwei Wang\thanks{Email: wanglw@cis.pku.edu.cn}}
\author[3]{Yuchen Zhou\thanks{Email: yuchen\_zhou@pku.edu.cn}}

\affil[1]{Key Laboratory of Machine Perception, School of EECS, Peking University}
\affil[2]{Department of Computer Science, Aarhus University}
\affil[3]{Key Laboratory of Machine Perception, School of Mathematics, Peking University}
\date{}
\maketitle

\abstract
Low rank matrix approximation is an important tool in machine learning. Given a data matrix, low rank approximation helps to find factors, patterns and provides concise representations for the data. Research on low rank approximation usually focus on real matrices. However, in many applications data are binary (categorical) rather than continuous. This leads to the problem of low rank approximation of binary matrix. Here we are given a $d \times n$ binary matrix $\mb{A}$ and a small integer $k$. The goal is to find two binary matrices $\mb{U}$ and $\mb{V}$ of sizes $d \times k$ and $k \times n$ respectively, so that the Frobenius norm of $\mb{A} - \mb{U} \mb{V}$ is minimized. There are two models of this problem, depending on the definition of the dot product of binary vectors: The $\mathrm{GF}(2)$ model and the Boolean semiring model. Unlike low rank approximation of real matrix which can be efficiently solved by Singular Value Decomposition, approximation of binary matrix is $\NP$-hard even for $k=1$.

In this paper, we consider the problem of Column Subset Selection (CSS), in which one low rank matrix must be formed by $k$ columns of the data matrix. We characterize the approximation ratio of CSS for binary matrices. For $\gf(2)$ model, we show the approximation ratio of CSS is bounded by $\frac{k}{2}+1+\frac{k}{2(2^k-1)}$ and this bound is asymptotically tight. For Boolean model, it turns out that CSS is no longer sufficient to obtain a bound. We then develop a Generalized CSS (GCSS) procedure in which the columns of one low rank matrix are generated from Boolean formulas operating bitwise on columns of the data matrix. We show the approximation ratio of GCSS is bounded by $2^{k-1}+1$, and the exponential dependency on $k$ is inherent.

\newpage
\fontsize{11pt}{13pt}\selectfont

\section{Introduction}
\label{Section:Introduction}
Low rank approximation of matrices is a classical problem. Given a matrix $\mb{A}$ of size $d \times n$, the goal is to find two low rank matrices $\mb{U}$ and $\mb{V}$, such that $\mb{U} \mb{V}$ approximates $\mb{A}$. Formally, the problem is to solve
\begin{equation}\label{Low_Rank_Problem}
\min_{\mb{U},\mb{V}} \|\mb{A}-\mb{U} \mb{V}\|^2_F,
\end{equation}
where the minimum is over all matrices $\mb{U}, \mb{V}$ of sizes $d \times k$ and $k \times n$ respectively; and $k$, typically a small integer, is the desired rank. Here the error is measured in terms of the Frobenius norm $ \| \cdot \|_F $.

In many applications, $\mb{A}$ is a data matrix. Each column of $\mb{A}$ is a $d$-dimensional data vector, and each row of $\mb{A}$ corresponds to an attribute. Low rank approximation of $\mb{A}$ is often called factor analysis and dimensionality reduction: the $k$ columns of the matrix $\mb{U}$ are the factors and basis vectors of the low dimensional space, and each column of $\mb{V}$ contains the combination coefficients.



If $\mb{A}$, $\mb{U}$, $\mb{V}$ are real matrices, low rank approximation can be efficiently solved by Singular Value Decomposition (SVD). This problem has been studied for more than a century, and is known as Principal Component Analysis (PCA) \cite{PCA:Hotelling33}, Karhunen-Lo\`eve Transform \cite{PCA:Karhunen47}, to name a few.

In this paper we consider low rank approximation of binary matrices. The motivation is that in many applications data are binary (categorical) rather than continuous. Indeed, nearly half data sets in UCI repository contains categorical features. In the binary case, we require both data matrix $\mb{A}$ and the rank-$k$ matrices $\mb{U}, \mb{V}$ are binary. There are two formulations of the binary low rank approximation problem, depending on the definition of vector dot product. One formulation will be referred to as $\mathrm{GF}(2)$ model, in which the dot product of two binary vectors $\mb{u}, \mb{v}$ is defined as $\mb{u}^T \mb{v} := \oplus_i u_i v_i$. The other formulation will be referred to as \emph{Boolean} model, in which the dot product is defined as $\mb{u}^T \mb{v} := \bigvee_i (u_i \wedge v_i)$.

The Boolean model is usually called Boolean Factor Analysis (BFA). It has found numerous applications in machine learning and data mining including latent variable analysis, topic models, association rule mining, clustering, and database tiling \cite{Belohlavek10, Miettinen08, Vaidya07, Seppanen03, Vsingliar06}. The $\mathrm{GF}(2)$ model has also been applied to Independent Component Analysis (ICA) over string data, attracting attention from the signal processing community \cite{Yeredor11, Gutch12, Painsky15}.




Despite of various applications and heuristic algorithms \cite{Miettinen08, Frolov07, Lucchese10, Frank12}, little is known on the theoretical side of the binary low rank approximation problem. In fact, previously the only known result is that for the very special case of $k=1$, for which the $\mathrm{GF}(2)$ and the Boolean model are equivalent, there are $2$-approximation algorithms (see Section \ref{subsection:related_works}).

In this paper, we provide the \emph{first} theoretical results for general binary low rank approximation problem, which is formally stated as follows. Given $\mb{A} \in \{0,1\}^{d \times n}$,
\begin{equation}\label{Binary_Low_Rank_Problem}
\min_{\mb{U} \in \{0,1\}^{d \times k}, \mb{V} \in \{0,1\}^{k \times n}} \|\mb{A} - \mb{U} \mb{V} \|^2_F.
\end{equation}
where the matrix product $\mb{U}\mb{V}$ is over $\mathrm{GF}(2)$ and Boolean semiring respectively.

Before stating the results, let us first see the differences between low rank approximation of real matrices and our $\gf(2)$ and Boolean models. First, the linear space over $\gf(2)$ has a very different structure from the Euclidean space. The dot product over $\gf(2)$ is not an inner product and does not induce a norm: There exists $\mb{a} \neq \mathbf{0}$ such that $\mb{a}^T \mb{a} = 0$ over $\gf(2)$. An immediate consequence is that for binary matrices of the $\gf(2)$ model, there is no Singular Value Decomposition (SVD), which is the basis for low rank approximation of real matrices. Therefore it is not clear how to obtain the optimal low rank approximation except for exhaustive search which requires $\Omega(2^{k\cdot \min(n,d)})$ time. The Boolean model is even more different: As it is a semiring rather than a field, we do not have a linear space (see below for details).

In words, there is no efficient algorithm that can solve the low rank approximation problem for binary matrices. In fact, we will show that finding the exact solution of (\ref{Binary_Low_Rank_Problem}) is $\NP$-hard even for $k=1$ (see Section \ref{Section:Hardness}). This result was obtained independently by Gillis and Vavasis~\cite{GillisV15}.



Another well-studied approach for low rank approximation of matrices is Column Subset Selection (CSS) \cite{Frieze04, Mahoney11randomized}. The goal of CSS is to find a subset of $k$ columns of $\mb{A}$ and form the low rank basis matrix so that the residual is as small as possible. An advantage of CSS is that the result is more interpretable than that of SVD. CSS has been extensively studied for low rank approximation of real matrices \cite{Gu96, Deshpande06, Deshpande06Matrix, Pan00, Drineas08, Boutsidis09, Tropp09, Deshpande10, Civril12, Boutsidis14, Cohen15, Yang15, Wang15, Bhaskara16}. Below is a formal definition of CSS over real matrices.

\begin{definition}[CSS for real matrices]
\label{Def:CSS}
Given a matrix $\mb{A} \in R^{d \times n}$ and a positive integer $k$, pick $k$ columns of $\mb{A}$ forming a matrix $\mb{P}_A \in R^{d \times k}$ such that the residual
\[
\|\mb{A} - \mb{P}_A\mb{Q}\|_{\xi}
\]
is minimized over all possible $\left(n \atop k \right)$ choices for the matrix $\mb{P}_A$. Here $\mb{Q}$ denotes the optimal matrix of size $k \times n$ given $\mb{P}_A$, which can be obtained by solving a least square problem; and $\xi = 2$ or $F$ denotes the spectral norm or Frobenious norm.
\end{definition}

The central problem in CSS is to determine the best function $\phi(n,k)$ on $n,k$ in the following bound.
\begin{equation}
\label{CSS_bound_form}
\|\mb{A} - \mb{P}_A \mb{Q}\|^2_{\xi} \le \phi(k,n) \|\mb{A} - \mb{A}_k\|^2_{\xi},
\end{equation}
where $\mb{A}_k$ denotes the best rank-$k$ approximation to the matrix $\mb{A}$ as computed with SVD.

Two classical results \cite{Gu96, Deshpande06Matrix} showed that for real matrices
\begin{equation}
\label{bound_spectral}
\|\mb{A} - \mb{P}_A \mb{Q}\|^2_{2} \le (k(n-k)+1) \|\mb{A} - \mb{A}_k\|^2_{2},
\end{equation}
and
\begin{equation}
\label{bound_Frob}
\|\mb{A} - \mb{P}_A \mb{Q}\|^2_{F} \le (k+1) \|\mb{A} - \mb{A}_k\|^2_{F}.
\end{equation}

There are extensive work on developing efficient algorithms for CSS with approximation ratio close to the above bounds, possibly using more than $k$ columns of $\mb{A}$. These include methods such as rank revealing QR \cite{Pan00}, adaptive sampling \cite{Deshpande06}, subspace sampling (leverage scores) \cite{Drineas08, Boutsidis09}, efficient volume sampling \cite{Deshpande10}, projection-cost preserving sketches \cite{Cohen15} and greedy CSS \cite{Bhaskara16}.

In this work, we study the CSS problem for binary matrices over $\gf(2)$ and Boolean semiring respectively. We ask the question in (\ref{CSS_bound_form}) and try to determine the best $\phi(k,n)$. Here we only consider Frobenious norm as the spectral norm does not exist in $\gf(2)$ or Boolean model.

The difficulty of the CSS problem for $\gf(2)$ and Boolean semiring model is that all methods developed for CSS over real matrices rely on at least one of the following concepts which are inherent in the Euclidean space: SVD, volume of a simplex, Eclidean distance, orthogonal projection, and QR decomposition. However, none of these concept exists in $\gf(2)$ or Boolean model.

In this paper, we develop new methods for the CSS problem for $\gf(2)$ and Boolean model respectively. For $\gf(2)$ model, we show that by picking the best $k$ columns of $\mb{A}$ to form $\mb{P}_A$,
\[
\|\mb{A} - \mb{P}_A \mb{Q}\|^2_{F} \le  \left(\frac{k}{2}+1+\frac{k}{2(2^k-1)} \right) \|\mb{A} - \mb{A}_k\|^2_{F},
\]
where $\mb{A}_k = \mb{U} \mb{V}$ is the optimal solution of (\ref{Binary_Low_Rank_Problem}). Moreover, we show the ratio $\left(\frac{k}{2}+1+\frac{k}{2(2^k-1)}\right)$ is asymptotically tight. Our technique is different from those for CSS over real matrices.



For Boolean model, it turns out that CSS is no longer sufficient for obtaining a bound, simply because the Boolean semiring does not have a field structure. We thus propose a Generalized CSS (GCSS) procedure. In this GCSS framework, we generate each column of the basis matrix $\mb{P}_A$ from carefully designed Boolean formulas operating bitwise on a predefined number of columns of $\mb{A}$. We show that GCSS achieves $(2^{k-1}+1)$-approximation ratio relative to $\|\mb{A} - \mb{A}_k\|^2_{F}$. Moreover, we argue that the exponential dependence in $k$ is inherent with the Boolean model (see Section \ref{Section:Boolean} for details).



Our work is a first step towards an understanding of low rank approximation of matrices over $\gf(2)$ and Boolean semiring. It is better to view our work as existence results for (Generalized) CSS for binary matrices, parallel to the classical existence theorems for CSS of real matrices given in (\ref{bound_spectral}) and (\ref{bound_Frob}) \cite{Gu96, Deshpande06Matrix}. Moreover, as SVD does not apply to $\gf(2)$ or Boolean models, CSS is so far the only method that obtains a low rank approximation for binary matrices with theoretical guarantees and deserves an in-depth study. Finally, it is an important future direction to develop efficient algorithms to achieve or approximately achieve the optimal ratio proved in this paper. We believe this requires new techniques exploiting the algebraic structure of $\gf(2)$ and Boolean semiring in a deep way.

The rest of this paper is organized as follows. In Section \ref{subsection:related_works} we discuss existing results on low rank approximation of binary matrices. In Section \ref{Section:GF(2)} we present the information-theoretically optimal upper bound for the approximation ratio of CSS over $\gf(2)$. In Section \ref{Section:Boolean} we propose the GCSS procedure and give the upper bound for the Boolean semiring model. In Section \ref{Section:Hardness} we show that finding the exaction solution of the low rank binary matrix approximation problem is $\NP$-hard even for $k=1$. Finally we conclude in Section \ref{section:conclusion}.

\subsection{Other Related Works}\label{subsection:related_works}

To the best of our knowledge, all known theoretical results on the low rank approximation problem are about the special case of rank-one, i.e., $k=1$. In the rank-one case, one looks for binary vectors $\mb{u}$, $\mb{v}$ such that $\|\mb{A}-\mb{u}\mb{v}^T\|_F$ is minimized, and therefore $\gf(2)$ and Boolean models are equivalent.

Shen et al. \cite{KDD:ShenJY09} formulate the rank-one problem as Integer Linear Programming (ILP). They showed that solving its LP relaxation yields a $2$-approximation. They also improved the efficiency by reducing the LP to a Max-Flow problem using a technique developed in \cite{Hochbaum97}. Jiang et al. \cite{Jiang14} observed that for the rank-one case, simply choosing the best column from $\mb{A}$ yields a $2$-approximation.


In the $\mathrm{GF}(2)$ model, low rank approximation is related to
the concept of matrix rigidity introduced by Valiant~\cite{Valiant77},
as a method of proving lower bounds for linear circuits. For a matrix
$\mb{A}$ over $\mathrm{GF}(2)$, the rigidity $R_{\mb{A}}(k)$ is the smallest
number of entries of $\mb{A}$ that must be changed in order to bring its
rank down to $k$. Thus for a $d \times n$ matrix $\mb{A}$, $R_{\mb{A}}(k)$ is
\emph{precisely} the minimum approximation error possible by a product
of a $d \times k$ matrix $\mb{U}$ and a $k \times n$ matrix $\mb{V}$. By the
results of Valiant, an $n \times n$ matrix $\mb{A}$ for which $R_{\mb{A}}(k) \geq
n^{1+\eps}$, for $k=O(n/\log\log n)$ and for some constant $\eps>0$
cannot be computed by a linear circuit of size $O(n)$ and depth
$O(\log n)$. Such rigid matrices exists in abundance -- the challenge
is to come up with an explicit construction of a family of rigid
matrices.  For the low rank approximation problem we are however
interested in the setting of $k \ll n$ and we are interested in
algorithms rather than explicit matrices.

\section{Column Subset Selection for Binary Matrices Over $\mathrm{GF}(2)$}\label{Section:GF(2)}

In this section we characterize the best possible approximation ratio of CSS in the $\gf(2)$ model. As mentioned in Section \ref{Section:Introduction}, the best approximation ratio of CSS for real matrices is $k+1$ under the Frobenius norm. This result is proved by the so-called volume sampling method  \cite{Deshpande06Matrix}. Concretely, the volume sampling method randomly samples a set of $k$  columns of $\mb{A}$ with probability proportional to the volume of the $k$-dimensional simplex formed by the $k$-columns along with the origin. Volume sampling generates an (expected) $k+1$ approximation ratio.

However, the $\gf(2)$ model does not have a notion of volume, since the dot product over $\gf(2)$ is not an inner product. Nevertheless, we develop a new approach and show the following bound.

\begin{theorem}\label{Thm:Main}
For any binary matrix $\mb{A} \in \{0,1\}^{d \times n}$, there exist $\mb{P}_A \in \{0,1\}^{d \times k}$ and $\mb{Q}\in \{0,1\}^{k \times n}$, where the columns of $\mb{P}_A$ are chosen from the columns of $\mb{A}$, such that
\[
\|\mb{A} - \mb{P}_A \mb{Q} \|^2_F \le \left(\frac{k}{2}+1+\frac{k}{2(2^k-1)}\right)\cdot \mathrm{OPT}_k,
\]
where
\[
\mathrm{OPT}_k := \|\mb{A} - \mb{A}_k \|^2_F,
\]
and $\mb{A}_k = \mb{U} \mb{V}$ is the optimal solution of (\ref{Binary_Low_Rank_Problem}). Here all matrix operations are over $\gf(2)$.
\end{theorem}


Moreover, we show that the approximation ratio $\left(\frac{k}{2}+1+\frac{k}{2(2^k-1)}\right)$ is asymptotically tight.\\

\begin{theorem}\label{Thm:Lower_Bound}
In the $\gf(2)$ model, for every $k \ge 1$ and every $\epsilon > 0$, there exists $\mb{A}$ such that
\[
\|\mb{A} - \mb{P}_A \mb{Q} \|^2_F > \left(\frac{k}{2}+1+\frac{k}{2(2^k-1)} - \epsilon \right)\cdot \mathrm{OPT}_k,
\]
for all $\mb{P}_A, \mb{Q}$, where $\mb{P}_A$ are formed by $k$ columns of $\mb{A}$.
\end{theorem}


Below, we give a high level description of the proof of the theorems. Our method uses the structure of $\gf(2)$ and is different to the techniques developed for CSS of real matrices.

Consider the problem given in (\ref{Binary_Low_Rank_Problem}). Throughout this paper, we will call the left matrix $\mb{U}$ the basis matrix, as its column vectors are the basis of the low dimensional space; and the right matrix $\mb{V}$ the coefficient matrix, as its columns contain the linear combination coefficients. Let $\mb{U}$ and $\mb{V}$ be an optimal solution of Eq.(\ref{Binary_Low_Rank_Problem}). Let $\mb{u_1},\ldots,\mb{u_k}$ be the $k$ columns of $\mb{U}$. For each column $\mb{u}_i$ of the optimal basis matrix $\mb{U}$, consider its nearest neighbor among all the columns of $\mb{A}$. Let $\mb{a}_1,\ldots,\mb{a}_n$ be the $n$ columns of $\mb{A}$. Denote by $\mb{a}_{(\mb{u}_i)}$ the nearest neighbor column of $\mb{u}_i$ in $\mb{A}$. Given an optimal basis matrix $\mb{U}$, we thus have a matrix $\mb{A}_{(\mb{U})}:=(\mb{a}_{(\mb{u}_1)},\ldots,\mb{a}_{(\mb{u}_k)})$, consisting of columns of $\mb{A}$. Note that the optimal solution of Eq.(\ref{Binary_Low_Rank_Problem}) is not unique. In fact, fixing an optimal basis matrix $\mb{U}$, for every matrix $\mb{B}=(\mb{b}_1,\ldots,\mb{b}_k)$, $\mb{b}_i \in \{0,1\}^k$, if the rank of $\mb{B}$ is $k$ over $\mathrm{GF}(2)$, then $(\mb{UB}, \mb{B}^{-1}\mb{V})$ must be an optimal solution also. (Throughout this section, matrix inverse and matrix rank are all over $\gf(2)$). Each optimal basis matrix $\mb{UB}$ induces a nearest neighbor matrix $\mb{A}_{(\mb{UB})}$. We will show that there must exist a rank $k$ matrix $\mb{B}$ such that the induced nearest neighbor matrix $\mb{A}_{(\mb{UB})}$, which when used as basis matrix, achieves an approximation error at most $(\frac{k}{2}+1+\frac{k}{2(2^k-1)})$ times that of the optimal solution $(\mb{UB},\mb{B}^{-1}\mb{V})$. Let $\mathrm{Err}(\mb{b}_1,\ldots,\mb{b}_k)$ be the approximation error associated with the basis matrix $\mb{A}_{(\mb{UB})}$ for $\mb{B}=(\mb{b}_1,\ldots,\mb{b}_k)$. Our goal is to bound the following.
\begin{equation}\label{Def_Err}
\min_{\mb{b}_1,\ldots,\mb{b}_k}\mathrm{Err}(\mb{b}_1,\ldots,\mb{b}_k),
\end{equation}
where $\mb{b}_i \in \{0,1\}^k$ for all $i \in [k]$.

Directly bounding Eq.(\ref{Def_Err}) is prohibitive. The approach we take is to consider a sequence of $k+1$ error minimization problems. For the $r$th ($0 \le r \le k$) minimization, we only optimize $r$ vectors among $\mb{b}_1,\ldots,\mb{b}_k$ and keep the other $k-r$ vectors fixed. Given $\mb{b}_1,\ldots,\mb{b}_k$, let
\begin{eqnarray}
\mathrm{Err}^{(0)}(\mb{b}_1,\ldots,\mb{b}_k) &:=& \mathrm{Err}(\mb{b}_1,\ldots,\mb{b}_k),\\
\mathrm{Err}^{(r)}(\mb{b}_1,\ldots,\mb{b}_{k-r}) &:=& \min_{\mb{b} \in \{0,1\}^k} \\& & \mathrm{Err}^{(r-1)}(\mb{b}_1,\ldots,\mb{b}_{k-r},\mb{b}), \nonumber \\
\mathrm{Err}^{(k)}() &:=& \min_{\mb{b} \in \{0,1\}^k} \mathrm{Err}^{(k-1)}(\mb{b}),
\end{eqnarray}
Then $\mathrm{Err}^{(k)}()$ is equivalent to Eq.(\ref{Def_Err}).

Although the final goal is to bound the \emph{ratio} between $\mathrm{Err}^{(k)}()$ and the error of the optimal solution of Eq.(\ref{Binary_Low_Rank_Problem}), we instead prove \emph{additive} bounds for $\mathrm{Err}^{(r)}(\mb{b}_1,\ldots,\mb{b}_{k-r})$ for all $0\le r \le k$. To be more precise, let $\mathrm{OPT}_k$ be the error of the optimal solution of Eq.(\ref{Binary_Low_Rank_Problem}), we will show that $\mathrm{Err}^{(r)}(\mb{b}_1,\ldots,\mb{b}_{k-r})$ is bounded by $\mathrm{OPT}_k$ plus a term depending on $r$ and $\mb{b}_1,\ldots,\mb{b}_{k-r}$ (Theorem \ref{Thm:Main_Techincal}). The point is that when $r=k$, this additive bound becomes a multiplicative bound with respect to $\mathrm{OPT}_k$ and is the ratio we want. The reason for introducing $\mathrm{Err}^{(0)},\ldots,\mathrm{Err}^{(k-1)}$ is that we need to use the relation between $\mathrm{Err}^{(r)}$ and $\mathrm{Err}^{(r-1)}$ to prove the bound. Actually the additive bound is proved by mathematical induction on $r$.

Although the relation of $\mathrm{Err}^{(r)}$ and $\mathrm{Err}^{(r-1)}$ is
\[
\mathrm{Err}^{(r)}(\mb{b}_1,\ldots,\mb{b}_{k-r}) = \min_{\mb{b}} \mathrm{Err}^{(r-1)}(\mb{b}_1,\ldots,\mb{b}_{k-r},\mb{b}),
\]
directly optimizing $\mb{b}$ is very difficult. Our idea is to use \emph{weighted averaging}. Since for each $\mb{b} \in \{0,1\}^k$,
\[
\mathrm{Err}^{(r)}(\mb{b}_1,\ldots,\mb{b}_{k-r}) \le \mathrm{Err}^{(r-1)}(\mb{b}_1,\ldots,\mb{b}_{k-r},\mb{b}),
\]
we have that for any set of weights $w_{\mb{b}}$ such that $w_{\mb{b}} \ge 0$ and $\sum_{\mb{b}} w_{\mb{b}} = 1$,
\[
\mathrm{Err}^{(r)}(\mb{b}_1,\ldots,\mb{b}_{k-r}) \le \sum_{\mb{b}} w_{\mb{b}} \mathrm{Err}^{(r-1)}(\mb{b}_1,\ldots,\mb{b}_{k-r},\mb{b}).
\]
We carefully choose the weights $w_{\mb{b}}$ to get a small upper bound. We conduct two layers of weighted averaging. Consider the quotient space $\mathrm{GF}(2)^k/\mathrm{span}(\mb{b}_1,\ldots,\mb{b}_{k-r})$ and the coset $[\mb{b}]:= \mb{b}+\mathrm{span}(\mb{b}_1,\ldots,\mb{b}_{k-r})$. In the first layer, we perform weighted averaging within each coset $[\mb{b}]$, and obtain a bound for $\mathrm{Err}^{(r)}$ depending on the coset. In the second layer we average over all cosets using another set of weights. We need different rules to set the weights in the two layers. Within a coset $[\mb{b}]$, we choose the weights as follows. Let $\mb{U}, \mb{V}$ be the already fixed optimal solution of Eq.(\ref{Binary_Low_Rank_Problem}). For each $\mb{c} \in [\mb{b}]$, let $n_{\mb{c}}$ denote the number of columns of $\mb{V}$ that are equal to $\mb{c}$. The weight we assign to $\mb{c}$ is proportional to $n_{\mb{c}}$. For the second layer, let
\[
n_{[\mb{b}]}:=\sum_{\mb{c} \in [\mb{b}]} n_{\mb{c}}
\]
be the total number of columns of $\mb{V}$ that belong to the coset $[\mb{b}]$. We assign the weight to a coset $[\mb{b}]$ as follows. If
\[
[\mb{b}] = \mathrm{span}(\mb{b}_1,\ldots,\mb{b}_{k-r}),
\]
then the weight is set to be zero. Otherwise, we assign the weight to $[\mb{b}]$ proportional to
\[
\frac{n_{[\mb{b}]}}{\sum_{[\mb{b}]} n_{[\mb{b}]} - \lambda n_{[\mb{b}]}},
\] 
where $\lambda$ is a constant depending on $r$. Combining the two layers of averaging we obtain the additive bound and that implies the desired approximation ratio. This finishes the description of the proof of Theorem \ref{Thm:Main}.

The lower bound in Theorem \ref{Thm:Lower_Bound} is proved by explicit construction. We construct a matrix which is approximately low rank in the sense that it is the product of two rank-$k$ matrix plus a very sparse matrix. The key ingredient of the proof is the construction of the two rank-$k$ matrices, which have special structures so that the approximation ratio of column subset selection cannot be smaller than $\frac{k}{2}+1+\frac{k}{2(2^k-1)}$ significantly.


The additive bounds are stated in Theorem \ref{Thm:Main_Techincal}, which is technical. Below we first describe the notions that will appear in Theorem \ref{Thm:Main_Techincal}. These notions will also be frequently used in the proof as well. For clarity, we list the notions in two tables.

\begin{definition}\label{Def:Notations}
For $1 \le r \le k$ and linear independent vectors $\mb{b}_{1} , \ldots
,\mb{b}_{r}$ in $\{0,1\}^{k}$:\\
\begin{table}[h]
  \centering
  \begin{tabular}{|l|l|}
    \hline
    Definition & Explanation\\
    \hline
    $\mathrm{span}^{c} (\mb{b}_{1} , \ldots ,\mb{b}_{r} ) := \{0,1\}^{k}$ & Complement of \\ $\setminus \mathrm{span} (\mb{b}_{1} , \ldots
    ,\mb{b}_{r} )$ & $\mathrm{span} (\mb{b}_{1}, \ldots ,\mb{b}_{r} )$.\\
    \hline
    $\mathrm{span}^{\backslash i} (\mb{b}_{1} , \ldots ,\mb{b}_{r} ) := \mathrm{span} ( $ & Span of all vectors except \\
    $\mb{b}_{1} , \ldots ,\mb{b}_{i-1}
    ,\mb{b}_{i+1} , \ldots ,\mb{b}_{r} )$ & the $i$th.\\
    \hline
  \end{tabular}
  \caption{Definitions for vector spans.}
\end{table}

Let $\mb{A}$ be the matrix to be approximated and $(\mb{U}, \mb{V})$ be a fixed optimal solution of the problem in Eq.(\ref{Binary_Low_Rank_Problem}). For $\mb{u} \in \zo^d$, $\mb{c} \in \{0,1\}^k$, and $\mathcal{X}
\subset \{0,1\}^{k}$:
\begin{table}[h]
  \centering
  \begin{tabular}{|l|l|}
    \hline
    Definition & Explanation\\
    \hline
    \multirow{4}{*}{$\mb{a}_{(\mb{u})}$} & The nearest neighbor of $\mb{u}$ \\ & among the columns of $\mb{A}$ (If \\&more than one nearest neigh-\\&bor, choose one arbitrarily.) \\
    \hline
    \multirow{2}{*}{$\mc{J}_{\mb{c}}:=\{j \in [n] :\mb{V}_{j} =\mb{c}\}$}  &
    The set of columns of $\mb{V}$ \\ & that are equal
    to vector $\mb{c}$.\\
    \hline
    \multirow{2}{*}{$n_{\mb{c}} := | \mc{J}_{\mb{c}} |$} & The
    number of columns of \\& $\mb{V}$ that are equal to $\mb{c}$.\\
    \hline
    \multirow{2}{*}{$L_{\mb{c}} := \sum_{j \in \mc{J}_{\mb{c}}} |\mb{a}_{j}
    -\mb{U}\mb{c}|$} & The total approximation \\& error of those columns in $\mc{J}_{\mb{c}}$.\\
    \hline
    \multirow{2}{*}{$N_{\mathcal{X}} := \sum_{\mb{c} \in \mathcal{X}} n_{\mb{c}}$} & The total number of columns \\& of $\mb{V}$ that belong to set
    $\mathcal{X}$.\\
    \hline
    \multirow{2}{*}{$M_{\mb{c}} = \left\{ \begin{array}{l}
     \frac{L_{\mb{c}}}{n_{\mb{c}}} \hspace{2em} n_{\mb{c}}>0\\
     d \hspace{2em} n_{\mb{c}} = 0
   \end{array}\right._{}$} & Upper bound of the average \\&
    error of the columns in $\mc{J}_{\mb{c}}$.\\
   \hline
  \end{tabular}
  \caption{Definitions for errors and nearest neighbors}
\end{table}

\end{definition}



Now we can state the additive bounds.

\begin{theorem}\label{Thm:Main_Techincal}
Let $\mb{b}_1,\ldots,\mb{b}_k$ be $k$ linear independent vectors in $\{0,1\}^k$. Then for each $0 \le r \le k$,

\begin{multline}\label{Additive_Bound}
\mathrm{Err}^{r}(\mb{b}_1,\ldots,\mb{b}_{k-r}) \le \mathrm{OPT}_k + \lambda_r \cdot \sum_{\mb{c} \in \mathrm{span}^c(\mb{b}_1,\ldots,\mb{b}_{k-r})}L_{\mb{c}} \\+ \sum_{i=1}^{k-r}f_i(\mb{b}_1,\ldots,\mb{b}_{k-r})M_{\mb{b}_i},
\end{multline}
where $M_{\mb{b}_i}$ has been defined in Definition \ref{Def:Notations}, and
  \[ \lambda_r = \left\{ \begin{array}{l}
       0 \hspace{7em} r = 0\\
       \frac{r}{2} \left( 1 + \frac{1}{2^r - 1} \right), \hspace{1em} 1 \leq r
       \leq k
     \end{array} \right._{} \]
and
\begin{equation}
f_i(\mb{b}_1,\ldots,\mb{b}_{k-r}) = N_{\mathcal{X}} + \frac{1}{2}N_{\mathcal{Y}},
\end{equation}
here
\[
\mathcal{X} = \mb{b}_i + \mathrm{span}^{\backslash i}(\mb{b}_1,\ldots,\mb{b}_{k-r}),
\]
and
\[
\mathcal{Y} = \mathrm{span}^c(\mb{b}_1,\ldots,\mb{b}_{k-r}).
\]
\end{theorem}

The formal proof of Theorem \ref{Thm:Main_Techincal} is lengthy and can be found in the supplementary materials. Theorem \ref{Thm:Main} follows from Theorem \ref{Thm:Main_Techincal} immediately.

\begin{proof}[Proof of Theorem \ref{Thm:Main}]

Let $r=k$ in Theorem \ref{Thm:Main_Techincal}. Then the last term in the RHS of Eq.(\ref{Additive_Bound}) vanishes. The second term in the RHS of Eq.(\ref{Additive_Bound}) becomes $\lambda_k \cdot \sum_{\mb{c} \in \{0,1\}^k} L_{\mb{c}}$. Observe that
\[
\sum_{\mb{c} \in \{0,1\}^k} L_{\mb{c}} = \mathrm{OPT}_k,
\]
and
\[
1+\lambda_k = \frac{k}{2}+1+\frac{k}{2(2^k-1)},
\]
the theorem follows.
\end{proof}

\section{Generalized CSS Over Boolean Semiring}\label{Section:Boolean}

It is not difficult to see that the method developed for $\gf(2)$ model in the previous section does not apply to the Boolean model, simply because Boolean semiring does not have a field structure. It turns out that, somewhat surprisingly, CSS is not sufficient to yield a bound relative to the optimal low rank solution in the Boolean model.

Here, we propose a Generalized CSS (GCSS) procedure. In GCSS, instead of using the columns of $\mb{A}$ directly to form $\mb{P}_A$, we perform carefully designed Boolean formulas (bitwise) to a predefined number of columns of $\mb{A}$ to form $\mb{P}_A$.

GCSS is rather involved. To illustrate the ideas, we first give an informal high level description of GCSS. We capture our GCSS by the following framework, which we denote as an
\emph{oblivious basis generation algorithm with advice}. Let $f(k)$
and $g(k)$ be functions of $k$. An oblivious basis generation
scheme with \emph{advice size} $f(k)$ and \emph{column dependence
  size} $g(k)$ operates as follows. On input input $\mb{a} \in
\zo^{f(k)}$ the scheme outputs $k$ Boolean formulas
$\Phi_1,\dots,\Phi_k$ each of $g(k)$ bits. Given $g(k)$ columns
$\mb{a}_{i_1},\dots,\mb{a}_{i_{g(k)}}$ of the matrix $\mb{A}$, the $k$ basis vectors
$\mb{u}_1,\dots,\mb{u}_k$ of $\mb{P}_A$ are constructed as
\[
\mb{u}_j = \Phi_j(\mb{a}_{i_1},\dots,\mb{a}_{i_{g(k)}}),
\]
where the Boolean function
$\Phi_j$ is applied \emph{entrywise}. From such a basis generation scheme we immediately obtain an
approximation result by iterating over all possible selections of
$g(k)$ columns of $\mb{A}$ as well as all possible advice strings $\mb{a} \in
\zo^{f(k)}$. We stress that the amount of information about $\mb{A}$ that can
be supplied to the algorithm using the advice string is
\emph{independent} of the actual size of $\mb{A}$. Our construction of GCSS will have column dependence size $2^k-1$ and advice size
$O(k2^k)$ in which we encode an \emph{ordering} of the given $2^k-1$
columns. This results in an approximation ratio of $2^{k-1}+1$.

To give a concrete description of GCSS, it is more convenient to use sets instead of vectors as the representation. For a column $\mb{a}_i$ of $\mb{A}$, let
\[
\mc{A}_i:= \{j \in [d]: (\mb{a}_{i})_j=1\},
\]
i.e., $\mb{a}_i$ is the characteristic string of $\mc{A}_i$. Similarly, for an optimal solution $(\mb{U,V})$ of the Boolean low rank approximation problem, let
\[
\mc{U}_i:= \{j \in [d]: (\mb{u}_{i})_j=1\},
\]
and
\[
\mc{V}_i:= \{j \in [k]: v_{ij}=1\}.
\]
Thus in this section we will always think of a column of $\mb{A}$, $\mb{U}$ or $\mb{V}$ as a set. Given a set $\mc{S} \subset [k]$, let
\[
\mc{J}_{\mc{S}}:=\{j \in [n]: \mc{V}_j = \mc{S} \},
\]
and $n_{\mc{S}}:= |\mc{J}_{\mc{S}}|$. Using these notions, the Boolean product of $\mb{U}$ and a vector which is the characteristic string of $\mc{S}$ will be denoted by $\mc{U}_{\mc{S}}:= \bigcup_{i \in \mc{S}} \mc{U}_i $. (Abuse the notion a little, we still use $\mc{U}_i$ instead of $\mc{U}_{\{i\}}$ from now on.) Like in the previous section, the nearest neighbor column of $\mc{U}_{\mc{S}}$ in $\mb{A}$ is defined by $\mb{a}_{(\mc{U}_{\mc{S}})}$. As we use set representation in this section, for notational simplicity we let $\mc{D}_{\mc{S}} \subset [d]$ be the set corresponding to this nearest neighbor column $\mb{a}_{(\mc{U}_{\mc{S}})}$, i.e., 
\[
\mc{D}_{\mc{S}} := \{i \in [d]: \mb{a}_{(\mc{U}_{\mc{S}})_i} = 1\}
\]


We are going to construct a rank-$k$ solution $\mc{B}_1,\ldots,\mc{B}_k$, where $\mc{B}_i \subset [d]$ is the set representation of the column of the basis matrix. Once the basis matrix is obtained, the coefficient matrix can be calculated in the same way as in the previous section. The concrete GCSS procedure is described in Algorithm 1.

Now we can state the main result of this section.

\begin{theorem}\label{Thm:Boolean}
GCSS (as described above) has approximation ratio $2^k$ relative to the optimal solution of (\ref{Binary_Low_Rank_Problem}) over Boolean semiring.
\end{theorem}

We now give the very high level idea of the proof. Fix a bijection
$\pi$ that satisfies $n_{S_1} \leq \dots \leq n_{S_{2^k}-1}$.  By
construction the set $\mc{D}_{S_\ell}$ is the best approximation to
$\mc{U}_{\mc{S}_{\ell}}$ given by a column of $\mb{A}$. Ideally the
sets $\mc{B}_1,\ldots,\mc{B}_k$ should be such that $\bigcup_{i \in
  S_\ell} \mc{B}_i$ is a comparable substitute for all $\ell$. What we
instead will be able to achieve is that for all $\ell \in [2^k-1]$
\begin{equation}\label{Boolean_Decomposition}
\mc{U}_{\mc{S}_\ell} \symdiff \left(\bigcup_{i \in S_\ell} \mc{B}_i \right)
\subseteq  \left(\bigcup_{\substack{\ell' \geq \ell}} \left( \mc{U}_{\mc{S}_{\ell'}} \symdiff \mc{D}_{\mc{S}_{\ell'}} \right)\right)
\end{equation}
where as seen from the algorithm the sets $\mc{B}_i$ are Boolean
combinations of the sets $\mc{D}_{\mc{S}_\ell}$. Intuively, we give
more importance to approximating the columns of $\mb{A}$ from
$\mc{J}_{\mc{S_\ell}}$ as $\ell$ increases. As the sizes
$n_{\mc{S}_\ell}$ of these sets of columns also increase this means
that we can account for the extra cost of possible poor approximation
of the sets $\mc{U}_{\mc{S}_{\ell}}$ for smaller $\ell$ in terms of
the approximation error of the sets $\mc{D}_{\mc{S}_{\ell'}}$ to
$\mc{U}_{\mc{S}_{\ell'}}$ for larger $\ell'\geq\ell$.

Intuitively we should attempt to approximate all the sets
$\mc{D}_{\mc{S}_\ell}$ simultaneously by $\bigcup_{i \in \mc{S}_\ell}
\mc{B}_i$. But since we work over a semiring we will have to work with
under-approximations. So for every $\ell$ we instead approximate the
under-approximation $\bigcup_{i\in \mc{S}_\ell} \mc{E}^\ell_i$ of
$\mc{D}_{\mc{S}_\ell}$. We do this by initially lettning $\mc{B}_i =
\mc{E}^1_i$ and then for each $\ell \in [2^k-1]$ adding
$\left(\bigcup_{i \in \mc{S}_\ell} \mc{E}^\ell_i \right) \setminus
\left(\bigcup_{i \in \mc{S}_\ell} \mc{E}^1_i \right)$ to $\bigcup_{i
  \in S_\ell} \mc{B}_i$. This last step has to be done carefully piece
by piece using the ordering of the sets
$\mc{S}_1,\dots,\mc{S}_{2^k-1}$. In the algorithm this is done using
the sets $\mc{F}^{\ell_1,\ell_2}_i$.

The approximation ratio of GCSS over Boolean semiring is $O(2^k)$, much larger than that of $\gf(2)$.  However we argue that this exponential dependency on $k$ is not an artifact of proof technique, it is inherent to the model.

\begin{algorithm}
	\label{alg:boolean_model}
	\caption{Generalized Column Subset Selection}
	\begin{algorithmic}[1]
		\FOR {all selection of $2^k-1$ column vectors $\mc{A}_{j_1},\mc{A}_{j_2},\ldots,\mc{A}_{j_{2^k-1}}$ in $\mb{A}$}
		\FOR {all bijections $\pi : [2^k-1] \rightarrow (2^{[k]} \setminus \{\emptyset\})$}
		
		\STATE Let $S_\ell = \pi(\ell)$ for $\ell \in [2^k-1]$\\
		\FOR {$i \in [k]$ and $\ell \in [2^k-1]$}
		\STATE Compute \\
		$$\mc{E}^\ell_i := \bigcap\limits_{\substack{\ell' \geq \ell :\\ i \in \mc{S}_{\ell'}}} D_{S_{l}}$$ where $D_\mc{S} = \mc{A}_{j_{\pi^{-1}(\mc{S})}}$ for $\emptyset \ne \mc{S} \subseteq [k]. $
		\ENDFOR
		
		\FOR {$1 \leq \ell_1 < \ell_2 \leq 2^k-1$ such that $i \in \mc{S}_{\ell_1} \cap \mc{S}_{l_2}$}
		\STATE Compute $$
		\mc{F}_i^{\ell_1,\ell_2} := \mc{E}_i^{\ell_1+1} \setminus \left[\bigcup_{i' \in \mc{S}_{\ell_2}} \mc{E}^{\ell_1}_{i'}\right].
		$$
		\ENDFOR
		\FOR{$i \in [k]$}
		\STATE Compute solution vector $\{\mc{B}_1,\mc{B}_2,\ldots,\mc{B}_k\}$by
		$$\mc{B}_i := \mc{E}_i^1 \cup \left(\bigcup_{\substack{\ell_1 < \ell_2:\\i \in \mc{S}_{\ell_1} \cap \mc{S}_{\ell_2}}} \mc{F}_i^{\ell_1,\ell_2} \right)\enspace .$$
		\ENDFOR
		\ENDFOR
		\STATE Compute the approximation error using the solution vector.
		\IF{the approximation error is optimal}
		\STATE Save $\{\mc{B}_1,\mc{B}_2,\ldots,\mc{B}_k\}$ as the output.
		\ENDIF
		\ENDFOR
	\end{algorithmic}
\end{algorithm}

Let $k$ be even and let $n=2^{k/2}$. We define the $n \times n$ matrix
$\mb{A}=(a_{\alpha,\beta})$ indexed by strings $\alpha,\beta \in
\zo^{k/2}$ by $a_{\alpha,\beta}=1$ if and only if $\alpha\neq
\beta$. Thus $\mb{A}$ is just the negation of the $n \times n$
identity matrix. It is well-known that the Boolean rank of $\mb{A}$ is
equal to $k$. In particular, we can write $\mb{A}$ as the boolean
product of $\mb{U}$ and $\mb{V}$, where the columns of $\mb{U}$ and
the rows of $\mb{V}$ are indexed by pairs $(i,b)$ where $i\in [k/2]$
and $b\in\{0,1\}$ and entry $(\alpha,(i,b))$ of $\mb{U}$ is $1$ if and
only if $\alpha_i=b$ and entry $((i,b),\beta)$ of $\mb{V}$ is $1$ if
and only if $\beta_i \neq b$. We note that the columns of $\mb{U}$ can
be written as Boolean formulas applied entry-wise to (all of) the
columns of $\mb{A}$. Since we consider approximation algorithms with
\emph{multiplicative error}, when supplied with input $A$ and $k$ our
algorithm is required to compute an \emph{exact factorization} of
$\mb{A}$ into $n \times k$ and $k \times n$ matrices $\mb{U}$ and
$\mb{V}$. If the underlying basis generation algorithm receives, say,
only half of the columns of $\mb{A}$ it does not seem possible to
compute such a factorization. It therefore seems that column
dependence size at least $2^{k/2-1}$ is necessary, which is about the
square-root of the column dependence size of our algorithm.

\begin{remark}
Using the technique of weighted averaging developed for the $\mathrm{GF}(2)$ model, we can actually improve the approximation ratio to $2^{k-1}+1$. We omit the details of the proof.
\end{remark}
The proof of Theorem \eqref{Thm:Boolean} can be found in supplementary materials.

\section{Hardness of Low Rank Approximation of Binary Matrices}\label{Section:Hardness}

Before our work, the computational complexity of the low rank approximation problem is not fully understood. For the rank-1 case, Tan showed that the equivalent problem \textsc{Maximum Edge Weight Biclique} for $\pmo$-matrices is \NP-hard under \textit{randomized} reductions~\cite{TAMC:Tan08}. In the case when the rank $k$ is unrestricted (i.e.\ part of the input) deciding whether there exist $\mb{U}$ and $\mb{V}$ such that $\mb{A}=\mb{U}\mb{V}$ in the Boolean semiring model is precisely the \NP-complete \textsc{Minimal Set Basis} problem~\cite{Stockmeyer95}, and that immediately implies that the approximation problem is \NP-hard to approximate wihtin \emph{any} factor, as noted by Miettinen~\cite{TKDE:MiettinenMGDM08}. On the other hand, this does not imply hardness when $k \ll d,n$. Indeed, the \textsc{Minimal Set Basis} problem is fixed-parameter tractable with parameter $k$, by a simple kernelization algorithm~\cite{TCS:FleischnerMPS09}. Note also that in the $\mathrm{GF}(2)$ model, deciding the existence of $\mb{U}$ and $\mb{V}$ such that $\mb{A}=\mb{U}\mb{V}$ is efficiently solvable using Gaussian elimination, regardless of the rank $k$ being unrestricted.

In this section we show the rank-$1$ \textsc{Binary Matrix Approximation} problem is \NP-hard under normal polynomial time reduction. We first define two related problems. Let $H$ be a \emph{complete} bipartite graph with edge weight, and let
$\mb{W}=(w_{ij})$ be the $d \times n$ matrix consisting of these edge
weights. The \textsc{Maximum Edge Weight Biclique} problem is to find a
biclique subgraph of $H$ with maximizing total edge weight. As an
optimization problem: maximize $\mb{x}^\transpose \mb{W} \mb{y}$, where $\mb{x} \in \zo^d$
and $\mb{y} \in \zo^n$.  The \textsc{Bipartite Max-Cut} problem is to find
a cut of the vertices of $H$ maximum the total weight of the edges
cut. As an optimization problem: maximize $\mb{x}^\transpose \mb{W} \mb{y}$, where $\mb{x}
\in \pmo^d$ and $\mb{y} \in \pmo^n$.  Note that these two problems differ
only in the domain from which $\mb{x}$ and $\mb{y}$ are chosen.

Shen, Ji, and Ye~\cite{KDD:ShenJY09} observed that the rank-1 \textsc{Binary Matrix Approximation}
problem is equivalent to \textsc{Maximum Edge Weight
  Biclique} when all edge weights are chosen from $\pmo$. Namely, if $\mb{A}$ is a $d \times n$ Boolean matrix, $\mb{u} \in
\zo^d$, and $\mb{v} \in \zo^n$, and let $\mb{J}_{d,n}$ denote the $d\times n$ all-$1$ matrix, we have
\begin{align*}
\norm{\mb{A}-\mb{uv}^\transpose}_F^2 & = \norm{\mb{A}}_F^2 - 2 \mb{u}^\transpose \mb{A v }+ \norm{\mb{u v}^\transpose }_F^2 \\&= \norm{\mb{A}}_F^2 - \mb{u}^\transpose(2\mb{A}-\mb{J}_{d,n})\mb{v}.
\end{align*}
Therefore, minimizing $\norm{\mb{A}-\mb{uv}^\transpose}_F^2$ is equivalent of maximizing $\mb{u}^\transpose(2\mb{A}-\mb{J}_{d,n})\mb{v}$. Also note that $(2\mb{A}-\mb{J}_{d,n})$ is a $\{-1,1\}$-matrix. Thus \NP-hardness of \textsc{Maximum Edge Weight Biclique} with $\pmo$ edge weights implies
\NP-hardness of rank-1 \textsc{Binary Matrix Approximation}. To show the  \NP-hardness of \textsc{Maximum Edge Weight Biclique}, we consider reduction from the \textsc{Bipartite Max-Cut} problem.

Roth and Viswanathan showed that Bipartite Max-Cut is \NP-hard even
when all wights are chosen from the set
$\pmo$~\cite{TIT:RothV08}. This is done by first showing \NP-hardness
when the weights are chosen from $\zpmo$ and then reducing to the case
of weights from $\pmo$.

Tan showed that \textsc{Maximum Edge Weight Biclique} is \NP-hard
\cite{TAMC:Tan08} when weights are chosen from $\zpmo$, and shows
\NP-hardness under randomized reductions when weights are chosen from
$\pmo$. He leaves it as an open problem to obtain \NP-hardness under
normal polynomial time reductions. The complexity of this problem was also stated as an open problem by Amit~\cite{Amit-MastersThesis}

The reduction from weights chosen from $\zpmo$ to $\pmo$ by Roth and
Viswanathan and by Tan is similar. The idea is to transform the $n
\times n$ $\zpmo$-weight matrix $\mb{W}$ into a new $nm \times nm$
$\pmo$-weight matrix $\mb{W}'$, where $\mb{W}'$ consists of $m \times m$ blocks
corresponding to each entry of $\mb{W}$. A $(-1)$-entry is transformed into
the all-$(-1)$ $m \times m$ matrix, and similarly is a $1$-entry
transformed into the all $1$ $m \times m$ matrix. But where Tan
transforms a $0$-entry to a \emph{random} $m \times m$ $\pmo$-matrix,
Roth and Viswanathan instead transforms a $0$-entry into a $m \times
m$ \emph{Hadamard} matrix. We will show that this transformation into \emph{Hadamard} matrix also
work in the setting of the \textsc{Maximum Edge Weight Biclique}
problem, thereby properly establishing its \NP-hardness.

\begin{theorem}\label{Thm:NP_Hardness}
The rank-$1$ \textsc{Binary Matrix Approximation} problem is \NP-hard.
\end{theorem}

The proof is a polynomial time many-one reduction from
\textsc{Maximum Edge Weight Biclique} with weights from $\zpmo$ to
\textsc{Maximum Edge Weight Biclique} with weights from $\zo$. Details can be found in supplementary materials.

\section{Conclusion}
\label{section:conclusion}

We study Column Subset Section (CSS) for low rank binary matrix approximation. CSS is often used as an alternative approach of SVD for low rank approximation of real matrices, where the advantage of CSS is the interpretability of its results. For binary matrices, CSS is so far the only approach with theoretical guarantee, as solving the low rank problem exactly is $\NP$-hard. We provide an upper bound on the approximation ration of CSS for the $\gf(2)$ model and show the bound is tight. This is a complete characterization from an information-theoretic point of view. For Boolean semiring model, we propose Generalized CSS (GCSS) since CSS is not strong enough to yield a bound in this scenario. We also show an upper bound for GCSS.

CSS has been actively studied for nearly three decades and the first work can at least date back to \cite{Golub65}, where it was called rank revealing QR in the numerical linear algebra community. The progress on CSS exhibits an interesting trajectory. Early results either gave bounds exponential in $k$ or the running time of the algorithm is $O(n^k)$ \cite{Foster86, Chan87, Hong92, Chan94, Chandrasekaran94, Bischof98}. By efforts from many researches, now there are polynomial time algorithms than have polynomial bound for the approximation ratio.

Our understanding of CSS for binary matrices is at the very beginning stage. It is an important future work to develop efficient CSS algorithms to achieve or approximately achieve the bounds in this paper.

	\section{Proof of Theorem \ref{Thm:Main_Techincal}}
	We are now going to prove Theorem \ref{Thm:Main_Techincal}. First we need a simple lemma.
	
	\begin{lemma}  \label{Lemma:Sum_Ineq}
		Let $a_1, \cdots, a_n$ and $\lambda$ be non-negative real numbers. Let $S := \sum_{i = 1}^n a_i$. If $S > \lambda a_i$ for all $i \in [n]$, then
		\begin{equation}\label{Sum_Ineq}
		\sum_{i = 1}^n \frac{a_i}{S - \lambda a_i} \geq \frac{n}{n - \lambda}.
		\end{equation}
	\end{lemma}
	
	\begin{proof}
		By Cauchy-Schwarz inequality, we have $\sum_{i = 1}^n (S - \lambda a_i) \sum_{i = 1}^n \frac{1}{S - \lambda a_i} \geq n^2$. Taking into consideration that $\sum_{i = 1}^n (S - \lambda a_i) = (n - \lambda) S$, we have $\sum_{i = 1}^n \frac{S}{S - \lambda a_i} \geq \frac{n^2}{n - \lambda}$. Observing that $\frac{S}{S - \lambda a_i} = 1 + \lambda \frac{a_i}{S - \lambda a_i}$, we further have $n + \lambda \sum_{i = 1}^n \frac{a_i}{S - \lambda a_i} \geq \frac{n^2}{n
			- \lambda}$. The Lemma follows from simple manipulations.
	\end{proof}
	
	We also note that $M_{\mb{c}}$ is an upper bound for the approximation error of $\mb{a}_{(\mb{Uc})}$ by $\mb{Uc}$, as stated in the following lemma.
	
	\begin{lemma}\label{Lemma:Upper_Bound_M}
		For any $\mb{c} \in \{0,1\}^k$, $|\mb{a}_{(\mb{Uc})}-\mb{Uc}| \le M_{\mb{c}}$
	\end{lemma}
	\begin{proof}
		If $n_{\mb{c}}=0$, then the lemma is true since $(\mb{a}_{(\mb{Uc})}-\mb{Uc}) \in \{0,1\}^d$. If $n_{\mb{c}}>0$, recall that $\mb{a}_{(\mb{Uc})}$ is the nearest neighbor of $\mb{Uc}$ among the columns of $\mb{A}$, and $L_{\mb{c}}$ is the total error of $n_{\mb{c}}$ columns. Therefore $|\mb{a}_{(\mb{Uc})}-\mb{Uc}| \le \frac{L_{\mb{c}}}{n_{\mb{c}}}$.
	\end{proof}
	
	\begin{proof} \textit{of~ Theorem \ref{Thm:Main_Techincal}}
		
		We prove the theorem by mathematical induction on $r$. Throughout the proof we will fix an optimal solution $(\mb{U},\mb{V})$ for the Binary low rank problem Eq.(\ref{Binary_Low_Rank_Problem}).\\
		
		\noindent\textbf{Base Case}\\
		
		We first prove inequality Eq.(\ref{Additive_Bound}) for $r=0$. Observe that in this case, $\lambda_0=0$ and $f_i(\mb{b}_1,\ldots,\mb{b}_k)=N_{\mb{b}_i + \mathrm{span}^{\backslash i}(\mb{b}_1,\ldots,\mb{b}_k)}$. Thus it suffices to show
		\[ \mathrm{Err}^{(0)} (\mb{b}_1, \cdots, \mb{b}_k) \leq
		\mathrm{OPT}_k + \sum_{i = 1}^k M_{\mb{b}_i} N_{\mb{b}_i +
			\mathrm{span}^{\backslash i} (\mb{b}_1, \cdots, \mb{b}_k)}. \]
		
		Recall that $\mathrm{Err}^{(0)} (\mb{b}_1, \cdots, \mb{b}_k)=\mathrm{Err} (\mb{b}_1, \cdots, \mb{b}_k)$, and $\mathrm{Err}(\mb{b}_1, \cdots, \mb{b}_k)$ is the approximation error of using $\mb{A}_{(\mb{U}\mb{B})}=(\mb{a}_{(\mb{Ub}_1)},\ldots,\mb{a}_{(\mb{Ub}_k)})$ as the basis matrix, where $\mb{a}_{(\mb{Ub}_i)}$ is the nearest neighbor of $\mb{Ub}_i$ among the columns of $\mb{A}$.
		
		The total approximation error $\mathrm{Err}(\mb{b}_1, \cdots, \mb{b}_k)$ is the sum of the error for each column $\mb{a}_j$. Let $\mb{b} \in \{ 0, 1 \}^k$. Because $\mb{b}_1, \cdots, \mb{b}_k$ are linear independent in $\mathrm{GF}(2)^k$, $\mb{b}$ can be represented by a linear combination of $\mb{b}_1, \cdots, \mb{b}_k$. Let $\mb{b}= \sum_{i = 1}^k  x_i \mb{b}_i$ for some $x_1,\ldots, x_k \in \{0,1\}$. (We abuse the notion a little bit, using $\sum$ for both ordinary addition and addition in $\mathrm{GF}(2)$. This should be always clear from the context.) Let $I_{\mb{b}} = \{i: x_i =1\}$ be the set that $\mb{b}_i$ contributes to $\mb{b}$. The approximation error of the column $\mb{a}_j$ can be written as $\min_{\mb{b}}|\mb{a}_j-\sum_{i \in I_{\mb{b}}}\mb{a}_{\mb{(Ub}_i)}|$. Recall that $\mc{J}_{\mb{c}}$ is the set of indices of the columns of $\mb{V}$ which are equal to $\mb{c}$. The partition $[n] = \cup_{\mb{c} \in \{0,1\}^k} \mc{J}_{\mb{c}}$ suggests a way to decompose the total approximation error. First, we have:
		\begin{eqnarray}
		 & & \sum_{j \in \mc{J}_{\mb{c}}} \left| \mb{a}_j - \sum_{i
			\in I_{\mb{c}}} \mb{a}_{(\mb{U}\mb{b}_i)} \right| \nonumber \\            
		& \leq & \sum_{j \in \mc{J}_{\mb{c}}} \left( \left| \mb{a}_j -
		\sum_{i \in I_{\mb{c}}} \mb{U}\mb{b}_i \right| + \sum_{i
			\in I_{\mb{c}}} | \mb{a}_{(\mb{U}\mb{b}_i)}
		-\mb{U}\mb{b}_i |\right) \nonumber\\
		& \leq & L_{\mb{c}} + n_{\mb{c}} \sum_{i \in
			I_{\mb{c}}} M_{\mb{b}_i},  \label{base_eq2}
		\end{eqnarray}
		where the last inequality is by the definition of $L_{\mb{c}}$, $n_{\mb{c}}$ and Lemma \ref{Lemma:Upper_Bound_M}. Recall that $L_{\mb{c}}$ is the sum of the approximation error for the optimal solution $(\mb{U},\mb{V})$ of the columns in $\mc{J}_{\mb{c}}$. Eq.(\ref{base_eq2}) leads to the following additive error bound.
		\begin{eqnarray}\label{r_0_ineq_1}
		&  & \mathrm{Err}^{(0)} (\mb{b}_1, \cdots, \mb{b}_k) \nonumber \\ 
        & = & \sum_{j =
			1}^n \min_{\mb{b} \in \{ 0, 1 \}^k} \left| \mb{a}_j
		- \sum_{i \in I_{\mb{b}}}
		\mb{a}_{(\mb{U}\mb{b}_i)} \right|\nonumber\\
		& = & \sum_{\mb{c} \in \{ 0, 1 \}^k} \sum_{j \in
			\mc{J}_{\mb{c}}} \min_{\mb{b} \in \{0,1\}^k}\left| \mb{a}_j - \sum_{i
			\in I_{\mb{b}}} \mb{a}_{(\mb{U}\mb{b}_i)}
		\right|\nonumber\\
		& \leq & \sum_{\mb{c} \in \{ 0, 1 \}^k} \sum_{j \in
			\mc{J}_{\mb{c}}} \left| \mb{a}_j - \sum_{i
			\in I_{\mb{c}}} \mb{a}_{(\mb{U}\mb{b}_i)}
		\right|\nonumber\\
		& \leq & \sum_{\mb{c} \in \{ 0, 1 \}^k} \left( L_{\mb{c}} +
		n_{\mb{c}} \sum_{i \in I_{\mb{c}}} M_{\mb{b}_i}
		\right)\nonumber\\
		& = & \mathrm{OPT}_k + \sum_{\mb{c} \in \{ 0, 1 \}^k} \sum_{i \in
			I_{\mb{c}}} n_{\mb{c}} M_{\mb{b}_i},
		\end{eqnarray}
		where the last inequality uses $\mathrm{OPT}_k = \sum_{\mb{c}}L_{\mb{c}}$. Consider the second term in Eq.(\ref{r_0_ineq_1}), we have
		\begin{eqnarray}\label{r_0_ineq_2}
		& &　\sum_{\mb{c} \in \{ 0, 1 \}^k} \sum_{i \in
			I_{\mb{c}}} n_{\mb{c}} M_{\mb{b}_i} \nonumber \\
            & = & \sum_{i = 1}^k \sum_{\mb{c} \in
			\{ 0, 1 \}^k} n_{\mb{c}} M_{\mb{b}_i} I [i \in
		I_{\mb{c}}] \nonumber \\
		& = &\sum_{i = 1}^k M_{\mb{b}_i} \left(
		\sum_{\mb{c} \in \mb{b}_i + \mathrm{span}^{\backslash i}
			(\mb{b}_1, \cdots, \mb{b}_{k - r})} n_{\mb{c}} \right)\nonumber \\
		& = & \sum_{i = 1}^k M_{\mb{b}_i} N_{\mb{b}_i
			+ \mathrm{span}^{\backslash i} (\mb{b}_1, \cdots, \mb{b}_{k -
				r})},
		\end{eqnarray}
		where the last inequality results from the definition of $N_{\mc{X}}$ in Definition \ref{Def:Notations}. Combining Eq.(\ref{r_0_ineq_1}) and Eq.(\ref{r_0_ineq_2}) finishes the proof for $r=0$.\\
		
		\noindent\textbf{Inductive Step}\\
		
		Assuming Eq.(\ref{Additive_Bound}) is true for all $r' \le r$, we are now going to show
		\fontsize{10pt}{13pt}\selectfont
		\begin{eqnarray}
		& & \mathrm{Err}^{(r + 1)} (\mb{b}_1, \cdots, \mb{b}_{k - r - 1}) \nonumber  \\ 
		&\leq & \mathrm{OPT}_k + \lambda_{r + 1} \sum_{\mb{c} \in
			\mathrm{span}^c_{} (\mb{b}_1, \cdots, \mb{b}_{k - r - 1})}
		L_{\mb{c}} + \nonumber \\
		& & \sum_{i = 1}^{k - r - 1} f_i (\mb{b}_1, \cdots,
		\mb{b}_{k - r - 1}) M_{\mb{b}_i}. \label{inductive_eq3}
		\end{eqnarray}
		\fontsize{11pt}{13pt}\selectfont
		Since \begin{align*}
		& & \mathrm{Err}^{(r + 1)} (\mb{b}_1, \cdots, \mb{b}_{k - r - 1}) \\ &=& \min_{\mb{b}} \mathrm{Err}^{(r)} (\mb{b}_1, \cdots, \mb{b}_{k - r - 1},  \mb{b})
		\end{align*}, we have for every set of weights $w_{\mb{b}}$ such that $w_{\mb{b}} \geq 0$ and $\sum_{\mb{b} \in \{ 0, 1  \}^k} w_{\mb{b}} = 1$,
		\begin{eqnarray}
		& & \mathrm{Err}^{(r + 1)} (\mb{b}_1, \cdots, \mb{b}_{k - r - 1}) \nonumber \\
		& \leq & \sum_{\mb{b} \in \{ 0, 1 \}^k} w_{\mb{b}}
		\mathrm{Err}^{(r)} (\mb{b}_1, \cdots, \mb{b}_{k - r - 1},
		\mb{b}). \label{inductive_eq4+}
		\end{eqnarray}
		We will conduct the weighted averaging in two layers. Consider the quotient space \\ $\mathrm{GF}(2)^k / \mathrm{span} (\mb{b}_1, \cdots, \mb{b}_{k  - r - 1})$ and the induced cosets. \\ Denote all the $2^{r+1}$ cosets by \\ $\mb{p}_0 + \mathrm{span} (\mb{b}_1, \cdots, \mb{b}_{k - r - 1})$ , $\cdots$ , \\ $\mb{p}_{2^{r + 1} - 1} + \mathrm{span} (\mb{b}_1, \cdots, \mb{b}_{k - r - 1})$. \\ Without loss of generality, assume $\mb{p}_0 \in \mathrm{span}  (\mb{b}_1, \cdots, \mb{b}_{k - r - 1})$. The first layer of weighted averaging will be performed within each coset $\mb{p}_{i} + \mathrm{span} (\mb{b}_1, \cdots, \mb{b}_{k - r - 1})$, and we will obtain an upper bound of $\mathrm{Err}^{(r + 1)} (\mb{b}_1, \cdots, \mb{b}_{k - r - 1})$ depending on $\mb{p}_i$. The second layer of averaging is over all the cosets, yielding the desired bound. The two layers use different rules for choosing the weights.\\
		
		\noindent\textbf{First Layer Weighted Averaging (within a coset): }\\
		
		For $\mb{p} \in \{\mb{p}_1,\ldots,\mb{p}_{2^{r+1}-1}\}$, define $Z(\mb{p}):=N_{\mb{p} + \mathrm{span} (\mb{b}_1, \cdots, \mb{b}_{k - r - 1})}$. So $Z(\mb{p})$ is the number of columns of $\mb{V}$ that belong to the coset indexed by $\mb{p}$. We will only consider those $\mb{p}$ such that $Z(\mb{p})>0$. Define the weights as $w_{\mb{b}} = n_{\mb{b}}/Z(\mb{p})$ if $\mb{b} \in \mb{p}+\mathrm{span} (\mb{b}_1, \cdots,\mb{b}_{k - r - 1})$ and zero otherwise. For every $\mb{b} \in \mb{p}+\mathrm{span} (\mb{b}_1, \cdots,\mb{b}_{k - r - 1})$ one has $\mathrm{span} (\mb{b}_1, \cdots,\mb{b}_{k - r - 1},\mb{b}) = \mathrm{span} (\mb{b}_1, \cdots,\mb{b}_{k - r - 1},\mb{p})$. Combining with the inductive hypothesis we have
		\begin{eqnarray}
		&  & \mathrm{Err}^{(r + 1)} (\mb{b}_1, \cdots, \mb{b}_{k - r -
			1}) \nonumber\\
		& \leq & \frac{1}{Z (\mb{p})} \sum_{\mb{b} \in \mb{p} + \mathrm{span}
			(\mb{b}_1, \cdots, \mb{b}_{k - r - 1})} n_{\mb{b}} \cdot \nonumber \\
		& &　\mathrm{Err}^{(r)} (\mb{b}_1, \cdots, \mb{b}_{k - r - 1},
		\mb{b}) \nonumber\\
		& \leq & \mathrm{OPT}_k + \lambda_r \sum_{\mb{c} \in \mathrm{span}^c_{}
			(\mb{b}_1, \cdots, \mb{b}_{k - r - 1}, \mb{p})}
		L_{\mb{c}} + \nonumber\\
		&  & \frac{1}{Z (\mb{p})} \sum_{\mb{b} \in \mb{p} + \mathrm{span}
			(\mb{b}_1, \cdots, \mb{b}_{k - r - 1})} n_{\mb{b}}　\cdot \nonumber \\
		& & [ f_{k - r} (\mb{b}_1, \cdots, \mb{b}_{k - r - 1},
		\mb{b}) M_{\mb{b}} + \nonumber\\
		& & \sum_{i = 1}^{k - r - 1} f_i
		(\mb{b}_1, \cdots, \mb{b}_{k - r - 1}, \mb{b})
		M_{\mb{b}_i} ] \nonumber\\
		&  & \label{inductive_eq5}
		\end{eqnarray}
		
		Our next goal is to bound the two terms in the last row of Eq.(\ref{inductive_eq5}) separately.\\
		
		\noindent\textbf{Bounding the first term in the last row of Eq.(\ref{inductive_eq5})}\\
		
		For the first term, we will show that
		\begin{eqnarray}\label{first_term}
		&  & \frac{1}{Z (\mb{p})} \sum_{\mb{b} \in \mb{p} + \mathrm{span}
			(\mb{b}_1, \cdots, \mb{b}_{k - r - 1})} n_{\mb{b}} \cdot \nonumber \\
		& & f_{k
			- r} (\mb{b}_1, \cdots, \mb{b}_{k - r - 1}, \mb{b})
		M_{\mb{b}}\nonumber\\
		& = &  \sum_{\mb{b} \in \mb{p} + \mathrm{span} (\mb{b}_1, \cdots,
			\mb{b}_{k - r - 1})} L_{\mb{b}} \cdot \nonumber \\
		& & \left( 1 + \frac{1}{2 Z (\mb{p})} N_{\mathrm{span}^c_{}
			(\mb{b}_1, \cdots, \mb{b}_{k - r - 1}, \mb{p})} \right).
		\end{eqnarray}
		By the definition of $M_{\mb{b}}$, it is easy to see $n_{\mb{b}} M_{\mb{b}} = L_{\mb{b}}$. \\ 
		So we focus on analyzing
		$f_{k - r} (\mb{b}_1, \cdots, \mb{b}_{k - r - 1}, \mb{b})$. \\ In fact, for every $ \mb{b} \in \mb{p} + \mathrm{span} (\mb{b}_1,
		\cdots, \mb{b}_{k - r - 1})$, we have
		\[\mathrm{span}^{\backslash (k - r)}  (\mb{b}_1, \cdots, \mb{b}_{k - r - 1}, \mb{b}) =
		\mathrm{span} (\mb{b}_1, \cdots, \mb{b}_{k - r - 1})\]
		Therefore
		\begin{align}
			 & \mb{b}+ \mathrm{span}^{\backslash (k - r)} (\mb{b}_1, \cdots, \mb{b}_{k - r - 1}, \mb{b}) \nonumber \\
			 =   & \mb{b}+ \mathrm{span}
			(\mb{b}_1, \cdots, \mb{b}_{k - r - 1}) \nonumber \\
			 =  &\mb{p} + \mathrm{span} (\mb{b}_1, \cdots, \mb{b}_{k - r - 1})
		\end{align}
		Also, \[\mathrm{span}^c_{} (\mb{b}_1, \cdots, \mb{b}_{k - r - 1},
		\mb{b}) = \mathrm{span}^c (\mb{b}_1, \cdots, \mb{b}_{k - r - 1}, \mb{p})\]
		Thus, for all $\mb{b} \in \mb{p} + \mathrm{span} (\mb{b}_1,
		\cdots, \mb{b}_{k - r - 1})$ we have
		\begin{eqnarray*}
			& &　f_{k - r} (\mb{b}_1, \cdots, \mb{b}_{k - r - 1}, \mb{b}) \\
			& = & N_{\mb{b}+ \mathrm{span}^{\backslash (k - r)} (\mb{b}_1,
				\cdots, \mb{b}_{k - r - 1}, \mb{b})} +  \\
                & & \frac{1}{2} 
			N_{\mathrm{span}^c_{} (\mb{b}_1, \cdots, \mb{b}_{k - r - 1},
				\mb{b})}\\
			& = & N_{\mb{p} + \mathrm{span} (\mb{b}_1, \cdots, \mb{b}_{k -
					r - 1})} + \frac{1}{2} N_{\mathrm{span}^c_{} (\mb{b}_1, \cdots,
				\mb{b}_{k - r - 1}, \mb{p})}
		\end{eqnarray*}
		Combining above, we have
		\begin{eqnarray*}
			&  & \frac{1}{Z (\mb{p})} \sum_{\mb{b} \in \mb{p} + \mathrm{span}
				(\mb{b}_1, \cdots, \mb{b}_{k - r - 1})} n_{\mb{b}} \cdot \\
			& & f_{k - r}
			(\mb{b}_1, \cdots, \mb{b}_{k - r - 1}, \mb{b})
			M_{\mb{b}} \nonumber\\
			& = & \sum_{\mb{b} \in \mb{p} + \mathrm{span}
				(\mb{b}_1, \cdots, \mb{b}_{k - r - 1})} L_{\mb{b}}
			\cdot \frac{1}{Z (\mb{p})} \cdot  \\ 
			& & \left( N_{\mb{p} + \mathrm{span}
				(\mb{b}_1, \cdots, \mb{b}_{k - r - 1})} + \frac{1}{2}
			N_{\mathrm{span}^c_{} (\mb{b}_1, \cdots, \mb{b}_{k - r - 1},
				\mb{p})} \right)  \nonumber\\
			& = &    \sum_{\mb{b} \in \mb{p} + \mathrm{span} (\mb{b}_1, \cdots,
				\mb{b}_{k - r - 1})} L_{\mb{b}} \cdot \\ 
			& &\left( 1 + \frac{1}{2 Z (\mb{p})} N_{\mathrm{span}^c_{}
				(\mb{b}_1, \cdots, \mb{b}_{k - r - 1}, \mb{p})} \right)  \label{eq6}
		\end{eqnarray*}
		This complete the proof of Eq.(\ref{first_term}).\\
		
		\noindent\textbf{Bounding the second term in the last row of Eq.(\ref{inductive_eq5})}\\
		
		For the second term, we are going to show
		\begin{eqnarray}\label{second_term}
		& & \frac{1}{Z (\mb{p})} \sum_{\mb{b} \in \mb{p} + \mathrm{span}
			(\mb{b}_1, \cdots, \mb{b}_{k - r - 1})} n_{\mb{b}}
		\sum_{i = 1}^{k - r - 1}M_{\mb{b}_i}  \cdot \nonumber \\
		& & f_i (\mb{b}_1, \cdots, \mb{b}_{k - r
			- 1}, \mb{b}) \nonumber \\
		& = & \frac{1}{Z (\mb{p})} \sum_{i = 1}^{k - r - 1} M_{\mb{b}_i}
		\sum_{\mb{b} \in \mb{p} + \mathrm{span} (\mb{b}_1, \cdots,
			\mb{b}_{k - r - 1})} n_{\mb{b}} \cdot \nonumber \\
		& & f_i (\mb{b}_1, \cdots,
		\mb{b}_{k - r - 1}, \mb{b}) \nonumber \\
		& \leq & \sum_{i = 1}^{k - r - 1} M_{\mb{b}_i} f_i (\mb{b}_1,
		\cdots, \mb{b}_{k - r - 1})
		\end{eqnarray}
		
		Consider $f_i (\mb{b}_1, \cdots, \mb{b}_{k - r - 1},  \mb{b})$. Observe that for all $1 \leq i \leq k - r -  1$ and   $\mb{b} \in \mb{p} + \mathrm{span} (\mb{b}_1, \cdots,
		\mb{b}_{k - r - 1})$, we have
		\begin{eqnarray*}
			& &　f_i (\mb{b}_1, \cdots, \mb{b}_{k - r - 1}, \mb{b}) 　\\
            & = &
			N_{\mb{b}_i + \mathrm{span}^{\backslash i} (\mb{b}_1, \cdots,
				\mb{b}_{k - r - 1}, \mb{b})} + \frac{1}{2} N_{\mathrm{span}^c_{}
				(\mb{b}_1, \cdots, \mb{b}_{k - r - 1}, \mb{b})}\\
			& = & N_{\mb{b}_i + \mathrm{span}^{\backslash i} (\mb{b}_1,
				\cdots, \mb{b}_{k - r - 1}, \mb{b})} + \frac{1}{2}
			N_{\mathrm{span}^c_{} (\mb{b}_1, \cdots, \mb{b}_{k - r - 1},
				\mb{p})}.
		\end{eqnarray*}
		Since
		\begin{eqnarray}
		& & \mb{b}_i + \mathrm{span}^{\backslash i} (\mb{b}_1, \cdots,
		\mb{b}_{k - r - 1}, \mb{b}) \nonumber \\ 
        & = & [\mb{b}_i +
		\mathrm{span}^{\backslash i} (\mb{b}_1, \cdots, \mb{b}_{k - r -
			1})] \cup \nonumber\\
		& &[\mb{b}_i +\mb{b}+ \mathrm{span}^{\backslash i} (_{}
		\mb{b}_1, \cdots, \mb{b}_{k - r - 1})]
		\end{eqnarray}
		
		we have
		\begin{eqnarray}
		& & f_i (\mb{b}_1, \cdots, \mb{b}_{k - r - 1}, \mb{b}) \nonumber\\ 
        & = &
		N_{\mb{b}_i +\mb{b}+ \mathrm{span}^{\backslash i}
			(\mb{b}_1, \cdots, \mb{b}_{k - r - 1})} +  \nonumber \\
		& & N_{\mb{b}_i +
			\mathrm{span}^{\backslash i} (\mb{b}_1, \cdots, \mb{b}_{k - r -
				1})} + \nonumber\\
		& & \frac{1}{2} N_{\mathrm{span}^c_{} (\mb{b}_1, \cdots,
			\mb{b}_{k - r - 1}, \mb{p})}. \label{second_term_eq8}
		\end{eqnarray}
		Let
		\begin{eqnarray}
		X_1 (i, \mb{p}) &=& \mb{p} + \mathrm{span}^{\backslash i} (\mb{b}_1,
		\cdots, \mb{b}_{k - r - 1}) \nonumber \\ X_2 (i, \mb{p}) &=& \mb{p}
		+\mb{b}_i + \mathrm{span}^{\backslash i} (\mb{b}_1, \cdots,
		\mb{b}_{k - r - 1}) \nonumber \\
		X_3 (i, \mb{p}) &=&\mb{b}_i + \mathrm{span}^{\backslash i}
		(\mb{b}_1, \cdots, \mb{b}_{k - r - 1}) \nonumber \\ X_4 (i, \mb{p}) &=&
		\mathrm{span}^c_{} (\mb{b}_1, \cdots, \mb{b}_{k - r - 1},
		\mb{p}) \label{second_term_eq9}
		\end{eqnarray}
		It is clear that $\forall j_1, j_2 \in \{ 1, 2, 3, 4 \}, j_1 \neq j_2$, one has \[ X_{j_1} (i, \mb{p})
		\cap X_{j_2} (i, \mb{p}) = \varnothing\] and \[  X_1 (i, \mb{p}) \cup X_2 (i, \mb{p}) = \mb{p} + \mathrm{span}
		(\mb{b}_1, \cdots, \mb{b}_{k - r - 1})\]
		If $\mb{b} \in X_1 (i, \mb{p}) $, then
		\begin{align*} & \mb{b}_i
		+\mb{b}+ \mathrm{span}^{\backslash i} (\mb{b}_1, \cdots,
		\mb{b}_{k - r - 1})\\
		=&\mb{b}_i + \mb{p} + \mathrm{span}^{\backslash
			i} (\mb{b}_1, \cdots, \mb{b}_{k - r - 1})\end{align*}Combining Eq.(\ref{second_term_eq8}) and Eq.(\ref{second_term_eq9}) yields
		\begin{align}
		 & f_i (\mb{b}_1, \cdots, \mb{b}_{k - r - 1}, \mb{b}) \nonumber \\
		=& N_{X_2 (i, \mb{p})} + N_{X_3 (i, \mb{p})} + \frac{1}{2} N_{X_4 (i,
			\mb{p})} \label{second_term_eq11}
		\end{align}
		If $\mb{b} \in X_2 (i, \mb{p})$,  then \begin{align*}
			& \mb{b}_i +\mb{b}+ \mathrm{span}^{\backslash i}
			(\mb{b}_1, \cdots, \mb{b}_{k - r - 1}) \\
			=& \mb{p} +
			\mathrm{span}^{\backslash i} (\mb{b}_1, \cdots, \mb{b}_{k - r -
				1})
		\end{align*}
		Combining Eq.(\ref{second_term_eq8}) and Eq.(\ref{second_term_eq9}) yields
		\begin{equation}
		f_i (\mb{b}_1, \cdots, \mb{b}_{k - r - 1}, \mb{b}) = N_{X_1 (i, \mb{p})} + N_{X_3 (i, \mb{p})} + \frac{1}{2} N_{X_4 (i,
			\mb{p})} \label{second_term_eq12}
		\end{equation}
		Now we bound the inner summation in the second line of Eq.(\ref{second_term}).
		\begin{eqnarray}
		&  & \sum_{\mb{b} \in \mb{p} + \mathrm{span} (\mb{b}_1, \cdots,
			\mb{b}_{k - r - 1})} n_{\mb{b}} f_i (\mb{b}_1, \cdots,
		\mb{b}_{k - r - 1}, \mb{b}) \nonumber\\
		& = & \sum_{\mb{b} \in X_1 (i, \mb{p}) \cup X_2 (i, \mb{p})}
		n_{\mb{b}} f_i (\mb{b}_1, \cdots, \mb{b}_{k - r - 1},
		\mb{b}) \nonumber\\
		& = & \sum_{\mb{b} \in X_1 (i, \mb{p})} n_{\mb{b}} f_i
		(\mb{b}_1, \cdots, \mb{b}_{k - r - 1}, \mb{b}) + \nonumber \\
		& & \sum_{\mb{b} \in X_2 (i, \mb{p})} n_{\mb{b}} f_i
		(\mb{b}_1, \cdots, \mb{b}_{k - r - 1}, \mb{b})
		\nonumber\\
		& = & \sum_{\mb{b} \in X_1 (i, \mb{p})} n_{\mb{b}} \left(
		N_{X_2 (i, \mb{p})} + N_{X_3 (i, \mb{p})} + \frac{1}{2} N_{X_4 (i,
			\mb{p})} \right) \nonumber\\
		&  & + \sum_{\mb{b} \in X_2 (i, \mb{p})}
		n_{\mb{b}} \left( N_{X_1 (i, \mb{p})} + N_{X_3 (i, \mb{p})} +
		\frac{1}{2} N_{X_4 (i, \mb{p})} \right) \nonumber\\
		& = & N_{X_1 (i, \mb{p})} \left( N_{X_2 (i, \mb{p})} + N_{X_3 (i,
			\mb{p})} + \frac{1}{2} N_{X_4 (i, \mb{p})} \right) \nonumber\\
		& & + N_{X_2 (i, \mb{p})}
		\left( N_{X_1 (i, \mb{p})} + N_{X_3 (i, \mb{p})} + \frac{1}{2} N_{X_4 (i,
			\mb{p})} \right) \nonumber\\
		& = & 2 N_{X_1 (i, \mb{p})} N_{X_2 (i, \mb{p})} + \nonumber \\ 
        & & (N_{X_1 (i, \mb{p})} +
		N_{X_2 (i, \mb{p})}) \left( N_{X_3 (i, \mb{p})} + \frac{1}{2} N_{X_4 (i,
			\mb{p})} \right) \nonumber\\
		& \leq & \frac{1}{2} (N_{X_1 (i, \mb{p})} + N_{X_2 (i, \mb{p})})^2 + \nonumber \\ & & 
		(N_{X_1 (i, \mb{p})} + N_{X_2 (i, \mb{p})}) \left( N_{X_3 (i, \mb{p})} +
		\frac{1}{2} N_{X_4 (i, \mb{p})} \right) \nonumber\\
		& = & (N_{X_1 (i, \mb{p})} + N_{X_2 (i, \mb{p})}) \cdot \nonumber \\ & & \left[ \frac{1}{2}
		(N_{X_1 (i, \mb{p})} + N_{X_2 (i, \mb{p})} + N_{X_4 (i, \mb{p})}) + N_{X_3
			(i, \mb{p})} \right]  \nonumber \\
            \label{second_term_eq13}
		\end{eqnarray}
		where the third equation follows from Eq.(\ref{second_term_eq11}) and Eq.(\ref{second_term_eq12}).
		By Eq.(\ref{second_term_eq9}) and Eq.(\ref{second_term_eq11}), we have
		\begin{equation}
		Z (\mb{p}) = N_{\mb{p} + \mathrm{span} (\mb{b}_1, \cdots,
			\mb{b}_{k - r - 1})} = N_{X_1 (i, \mb{p})} + N_{X_2 (i, \mb{p})},
		\label{second_term_eq14}
		\end{equation}
		and
		\begin{equation}
		X_1 (i, \mb{p}) \cup X_2 (i, \mb{p}) \cup X_4 (i, \mb{p}) =
		\mathrm{span}^c_{} (\mb{b}_1, \cdots, \mb{b}_{k - r - 1}).
		\label{eq15}
		\end{equation}
		Thus
		\begin{equation}
		N_{X_1 (i, \mb{p})} + N_{X_2 (i, \mb{p})} + N_{X_4 (i, \mb{p})} =
		N_{\mathrm{span}^c_{} (\mb{b}_1, \cdots, \mb{b}_{k - r - 1})}.
		\label{second_term_eq16}
		\end{equation}
		Therefore, by Eq.(\ref{second_term_eq9}), Eq.(\ref{second_term_eq14}), Eq.(\ref{second_term_eq16}) and the definition of $f_i(\mb{b}_1, \cdots, \mb{b}_{k - r - 1})$ we have
		\begin{eqnarray}
		&  & (N_{X_1 (i, \mb{p})} + N_{X_2 (i, \mb{p})}) \cdot \nonumber \\ & & \left[ \frac{1}{2}
		(N_{X_1 (i, \mb{p})} + N_{X_2 (i, \mb{p})} + N_{X_4 (i, \mb{p})}) + N_{X_3
			(i, \mb{p})} \right] \nonumber\\
		& = & Z (\mb{p}) \cdot ( \frac{1}{2} N_{\mathrm{span}^c_{} (\mb{b}_1,
			\cdots, \mb{b}_{k - r - 1})} + \nonumber \\ & & N_{\mb{b}_i + 
			\mathrm{span}^{\backslash i} (\mb{b}_1, \cdots, \mb{b}_{k - r -
				1})} ) \nonumber\\
		& = & Z (\mb{p}) f_i (\mb{b}_1, \cdots, \mb{b}_{k - r - 1}).
		\label{second_term_eq17}
		\end{eqnarray}
		Substitute Eq.(\ref{second_term_eq17}) into Eq.(\ref{second_term_eq13}), we get
		\begin{align}
		 & & \sum_{\mb{b} \in \mb{p} + \mathrm{span} (\mb{b}_1, \cdots,
			\mb{b}_{k - r - 1})} n_{\mb{b}} f_i (\mb{b}_1, \cdots,
		\mb{b}_{k - r - 1}, \mb{b}) \nonumber \\
        &\leq& Z (\mb{p}) f_i
		(\mb{b}_1, \cdots, \mb{b}_{k - r - 1}).
		\end{align}
		Now we are able to prove Eq.(\ref{second_term}),
		\begin{eqnarray}
		&  & \frac{1}{Z (\mb{p})} \sum_{\mb{b} \in \mb{p} + \mathrm{span}
			(\mb{b}_1, \cdots, \mb{b}_{k - r - 1})} n_{\mb{b}} \cdot \nonumber \\
		& & \sum_{i = 1}^{k - r - 1} f_i (\mb{b}_1, \cdots, \mb{b}_{k - r
			- 1}, \mb{b}) M_{\mb{b}_i} \nonumber\\
		& = & \frac{1}{Z (\mb{p})} \sum_{i = 1}^{k - r - 1} M_{\mb{b}_i}
		\sum_{\mb{b} \in \mb{p} + \mathrm{span} (\mb{b}_1, \cdots,
			\mb{b}_{k - r - 1})} n_{\mb{b}} \cdot \nonumber \\ & & f_i (\mb{b}_1, \cdots,
		\mb{b}_{k - r - 1}, \mb{b}) \nonumber\\
		& \leq & \frac{1}{Z (\mb{p})} \sum_{i = 1}^{k - r - 1} M_{\mb{b}_i}
		Z (\mb{p}) f_i (\mb{b}_1, \cdots, \mb{b}_{k - r - 1})
		\nonumber\\
		& = & \sum_{i = 1}^{k - r - 1} M_{\mb{b}_i} f_i (\mb{b}_1,
		\cdots, \mb{b}_{k - r - 1}). \label{second_term_eq18}
		\end{eqnarray}
		This finishes bounding the second term in Eq.(\ref{inductive_eq5}).

		Finally, combine Eq.(\ref{inductive_eq5}), Eq.(\ref{first_term}) and Eq.(\ref{second_term}), we have for all $\mb{p}
		\in \{ \mb{p}_1, \cdots, \mb{p}_{2^{r + 1} - 1} \}$ such that $Z (\mb{p}) > 0$
		\begin{eqnarray}\label{First_Layer_I}
		&  & \mathrm{Err}^{(r + 1)} (\mb{b}_1, \cdots, \mb{b}_{k - r -
			1}) \label{eq19} \nonumber\\
		& \leq & \mathrm{OPT}_k + \lambda_r \sum_{\mb{c} \in \mathrm{span}^c_{}
			(\mb{b}_1, \cdots, \mb{b}_{k - r - 1}, \mb{p})}
		L_{\mb{c}} \nonumber \\& & + \sum_{i = 1}^{k - r - 1} f_i (\mb{b}_1, \cdots,
		\mb{b}_{k - r - 1}) M_{\mb{b}_i} \nonumber\\
		&  & + \left( 1 + \frac{1}{2 Z (p)} N_{\mathrm{span}^c_{}
			(\mb{b}_1, \cdots, \mb{b}_{k - r - 1}, \mb{p})} \right) \cdot \nonumber \\ & &
		\sum_{\mb{b} \in \mb{p} + \mathrm{span} (\mb{b}_1, \cdots,
			\mb{b}_{k - r - 1})} L_{\mb{b}}.
		\end{eqnarray}
		Note that $\mathrm{span}^c_{} (\mb{b}_1, \cdots, \mb{b}_{k - r - 1}) =
		\mathrm{span}^c_{} (\mb{b}_1, \cdots, \mb{b}_{k - r - 1}, \mb{p})
		\cup [\mb{p} + \mathrm{span} (\mb{b}_1, \cdots, \mb{b}_{k - r -
			1})]$, we have
		\begin{eqnarray}\label{Minus}
		&&\sum_{\mb{c} \in \mathrm{span}^c_{} (\mb{b}_1, \cdots,
			\mb{b}_{k - r - 1}, \mb{p})} L_{\mb{c}} \nonumber \\ &=& \sum_{\mb{c}
			\in \mathrm{span}^c_{} (\mb{b}_1, \cdots, \mb{b}_{k - r - 1})}
		L_{\mb{c}} \nonumber \\ &-& \sum_{\mb{b} \in \mb{p} + \mathrm{span}
			(\mb{b}_1, \cdots, \mb{b}_{k - r - 1})} L_{\mb{b}}.
		\end{eqnarray}
		Substitute Eq.(\ref{Minus}) into Eq.(\ref{First_Layer_I}), we obtain
		\begin{eqnarray}
		&  & \mathrm{Err}^{(r + 1)} (\mb{b}_1, \cdots, \mb{b}_{k - r -
			1}) \nonumber\\
		& \leq & \mathrm{OPT}_k + \lambda_r \sum_{\mb{c} \in \mathrm{span}^c_{}
			(\mb{b}_1, \cdots, \mb{b}_{k - r - 1})} L_{\mb{c}} + \nonumber\\
            & & 
		\sum_{i = 1}^{k - r - 1} f_i (\mb{b}_1, \cdots, \mb{b}_{k - r
			- 1}) M_{\mb{b}_i} +  \nonumber\\
		&  &  \left( 1 + \frac{1}{2 Z (\mb{p})} N_{\mathrm{span}^c_{}
			(\mb{b}_1, \cdots, \mb{b}_{k - r - 1}, \mb{p})}-\lambda_r \right) \cdot \nonumber \\
            & & 
		\sum_{\mb{b} \in p + \mathrm{span} (\mb{b}_1, \cdots,
			\mb{b}_{k - r - 1})} L_{\mb{b}}. \label{First_Layer_II}
		\end{eqnarray}
		Now we finish the first layer weighted averaging. Eq.(\ref{First_Layer_II}) is the bound obtained by averaging within the coset indexed by $\mb{p}$. This bound will be further used in the second layer.
		
		Before we move to the second layer averaging, we point out that the inductive step $r \rightarrow r+1$ for the special case $r = 0$ has already been proved; and we do not need to involve the second layer. To see this, note that when $r=0$, $\mathrm{span}^c_{} (\mb{b}_1, \cdots, \mb{b}_{k - r  - 1}, \mb{p}) = \emptyset$ and $\mb{p} + \mathrm{span} (\mb{b}_1, \cdots,  \mb{b}_{k - r - 1}) = \mathrm{span}^c (\mb{b}_1, \cdots,  \mb{b}_{k - r - 1})$. Substitute these two equalities into Eq.(\ref{First_Layer_I}) yields
		\begin{eqnarray*}
			& & \mathrm{Err}^{(1)} (\mb{b}_1, \cdots, \mb{b}_{k - 1}) \\
			& = & \mathrm{OPT}_k + \sum_{i = 1}^{k - 1} f_i (\mb{b}_1, \cdots,
			\mb{b}_{k - 1}) M_{\mb{b}_i} + \\ 
            & &  \sum_{\mb{b} \in
				\mathrm{span}^c (\mb{b}_1, \cdots, \mb{b}_{k - 1})}
			L_{\mb{b}}\\
			& = & \mathrm{OPT}_k + \sum_{i = 1}^{k - 1} f_i (\mb{b}_1, \cdots,
			\mb{b}_{k - 1}) M_{\mb{b}_i} + \\
            & & \lambda_{1}
			\sum_{\mb{b} \in \mathrm{span}^c (\mb{b}_1, \cdots,
				\mb{b}_{k - 1})} L_{\mb{b}},
		\end{eqnarray*}
		where the last step uses $\lambda_{r + 1} = \lambda_1 = 1$.
		
		Now we will move to the second layer averaging, and we will assume $r \ge 1$ since $r=0$ has already been proved.\\
		
		\noindent\textbf{Second Layer Weighted Averaging (over the cosets)}\\
		
		In this layer we will conduct weighted averaging over all the nontrivial cosets. A coset $\mb{p}_i+\mathrm{span}(\mb{b}_1,\ldots,\mb{b}_{k-r-1})$ is nontrivial if $Z(\mb{p}_i)>0$ and $i \neq 0$ (i.e., the coset is not equal to \\ $\mathrm{span}(\mb{b}_1,\ldots,\mb{b}_{k-r-1})$). In Eq.(\ref{First_Layer_II}) we already obtain an upper bound for $\mathrm{Err}^{r+1}$ by weighted averaging within each nontrivial coset. In this layer we will prove the theorem by averaging over all the nontrivial cosets based on Eq.(\ref{First_Layer_II}). Note that in the upper bound Eq.(\ref{First_Layer_II}), only the last term depends on the coset. So we will focus on this term. For notational simplicity, denote this last term as $H$:
		\begin{eqnarray}\label{last_term_H}
		H &:=& \left( 1 + \frac{1}{2 Z (\mb{p})} N_{\mathrm{span}^c_{}
			(\mb{b}_1, \cdots, \mb{b}_{k - r - 1}, \mb{p})}-\lambda_r \right) \cdot \nonumber \\
		& & \sum_{\mb{b} \in \mb{p} + \mathrm{span} (\mb{b}_1, \cdots,
			\mb{b}_{k - r - 1})} L_{\mb{b}}.
		\end{eqnarray}
		Note that
		\begin{eqnarray}
         &[\mb{p}_i + \mathrm{span} (\mb{b}_1, \cdots, \mb{b}_{k - r -
			1})]\nonumber \\
             \cup& \mathrm{span}^c (\mb{b}_1, \cdots, \mb{b}_{k - r -
			1}, \mb{p}_i) \nonumber \\ =& \mathrm{span}^c (\mb{b}_1, \cdots, \mb{b}_{k -
			r - 1}), 
        \end{eqnarray}
		and
		\begin{eqnarray}
        &\bigcup_{i = 1}^{2^{r + 1} - 1} \mb{p}_i + \mathrm{span} (\mb{b}_1,
		\cdots, \mb{b}_{k - r - 1}) \nonumber \\ =& \mathrm{span}^c (\mb{b}_1,
		\cdots, \mb{b}_{k - r - 1}). 
        \end{eqnarray}
        We have
		\begin{eqnarray}& &  N_{\mathrm{span}^c (\mb{b}_1, \cdots, \mb{b}_{k - r - 1}, \mb{p}_i)} \nonumber \\ &=& N_{\mathrm{span}^c (\mb{b}_1, \cdots, \mb{b}_{k -
				r - 1})} - N_{\mb{p}_i + \mathrm{span} (\mb{b}_1, \cdots,
			\mb{b}_{k - r - 1})} \nonumber \\ &=& \sum_{j = 1}^{2^{r + 1} - 1} Z (\mb{p}_j) -
		Z (\mb{p}_i). 
        \end{eqnarray}
		Thus for a nontrivial coset indexed by $\mb{p}_i$ we have
		\begin{eqnarray}
		& &1 + \frac{1}{2 Z (\mb{p}_i)} N_{\mathrm{span}^c_{} (\mb{b}_1, \cdots,
			\mb{b}_{k - r - 1}, \mb{p}_i)} - \lambda_r \nonumber \\ &=& \frac{1}{2 Z
			(\mb{p}_i)} \left[ \sum_{j = 1}^{2^{r + 1} - 1} Z (\mb{p}_j) - (2
		\lambda_r - 1) Z (\mb{p}_i) \right] \nonumber \\ \label{Second_Layer_eq21}
		\end{eqnarray}
		By Eq.(\ref{last_term_H}) and Eq.(\ref{Second_Layer_eq21}), we have that for all $i \in [2^{r + 1} - 1]$ such that $Z
		(\mb{p}_i) > 0$,
		\begin{eqnarray}
		H &=& \frac{1}{2 Z (\mb{p}_i)} \left[ \sum_{j = 1}^{2^{r + 1} - 1} Z
		(\mb{p}_j) - (2 \lambda_r - 1) Z (\mb{p}_i) \right] \nonumber \\ & & \sum_{\mb{b} \in
			\mb{p}_i + \mathrm{span} (\mb{b}_1, \cdots, \mb{b}_{k - r - 1})}
		L_{\mb{b}} \label{Second_Layer_eq23}
		\end{eqnarray}
		
		Before we start the weighted averaging, we need to treat a special case differently. Consider the case that there exists a nontrivial coset indexed by $\mb{p}_i$ such that\\
		$\sum_{j =  1}^{2^{r + 1} - 1} Z (\mb{p}_j) - (2 \lambda_r - 1) Z (\mb{p}_i) \leq 0$. In this case, $H \le 0$ as a result of Eq.(\ref{Second_Layer_eq23}). Combining Eq.(\ref{First_Layer_II}) and the fact that $\lambda_r \le \lambda_{r+1}$, we have
		\begin{eqnarray*}
			& & \mathrm{Err}^{(r + 1)} (\mb{b}_1, \cdots, \mb{b}_{k - r - 1}) \\
			& \leq & \mathrm{OPT}_k + \lambda_r \sum_{\mb{c} \in \mathrm{span}^c_{}
				(\mb{b}_1, \cdots, \mb{b}_{k - r - 1})} L_{\mb{c}} + \\ & & 
			\sum_{i = 1}^{k - r - 1} f_i (\mb{b}_1, \cdots, \mb{b}_{k - r
				- 1}) M_{\mb{b}_i}\\
			& \leq & \mathrm{OPT}_k + \lambda_{r + 1} \sum_{\mb{c} \in
				\mathrm{span}^c_{} (\mb{b}_1, \cdots, \mb{b}_{k - r - 1})}
			L_{\mb{c}} +  \\ 
            & & \sum_{i = 1}^{k - r - 1} f_i (\mb{b}_1, \cdots,
			\mb{b}_{k - r - 1}) M_{\mb{b}_i}
		\end{eqnarray*}
		So the theorem is true for this special case.
		
		Below we conduct the weighted averaging assuming all the nontrivial cosets satisfy
		\[\sum_{j = 1}^{2^{r +  1} - 1} Z (\mb{p}_j) - (2 \lambda_r - 1) Z (\mb{p}_i) > 0.\]
		The weights are set as follows. We only give positive weights to nontrivial cosets. We assign to a nontrivial coset indexed by $\mb{p}_i$ weight $w_i$ given by:
		\[ w_i = \frac{2 Z (\mb{p}_i)}{\sum_{j = 1}^{2^{r + 1} - 1} Z (\mb{p}_j) -
			(2 \lambda_r - 1) Z (\mb{p}_i)}. \]
		By Eq.(\ref{Second_Layer_eq23}) we thus have
		\[ w_iH = \sum_{\mb{b} \in \mb{p}_i + \mathrm{span}
			(\mb{b}_1, \cdots, \mb{b}_{k - r - 1})} L_{\mb{b}}. \]
		Hence
		\begin{eqnarray}
		\sum_{i = 1}^{2^{r + 1} - 1} w_i H &=& \sum_{i = 1}^{2^{r + 1} - 1}
		\sum_{\mb{b} \in \mb{p}_i + \mathrm{span} (\mb{b}_1, \cdots,
			\mb{b}_{k - r - 1})} L_{\mb{b}} \nonumber \\ &=& \sum_{\mb{b} \in \mathrm{span}^c_{}
			(\mb{b}_1, \cdots, \mb{b}_{k - r - 1})} L_{\mb{b}}. \label{Second_Layer_eq24}
		\end{eqnarray}
		On the other hand, by Lemma \ref{Lemma:Sum_Ineq} we have
		\begin{eqnarray}
		& & \sum_{i = 1}^{2^{r + 1} - 1} w_i \nonumber \\
        & = & 2 \sum_{i = 1}^{2^{r + 1} - 1}
		\frac{Z (\mb{p}_i)}{\sum_{j = 1}^{2^{r + 1} - 1} Z (\mb{p}_j) - (2
			\lambda_r - 1) Z (\mb{p}_i)} \nonumber\\
		& \geq & \frac{2(2^{r + 1} - 1)}{2^{r + 1} - 1 - (2 \lambda_r - 1)}
		\nonumber\\
		& = & \frac{2^{r + 1} - 1}{2^r - \lambda_r} \label{Second_Layer_eq25}
		\end{eqnarray}
		Observe that  $(2^r - 1) \lambda_r = r 2^{r - 1}$, we thus have
		
		\[(2^{r + 1} - 1)
		\lambda_{r + 1} - 2 (2^r - 1) \lambda_r = (r + 1) 2^r - r 2^r = 2^r,\]
		and therefore,
		\begin{equation}\label{Second_Layer_lambda}
		(2^{r + 1} - 1) \lambda_{r + 1} - (2^{r + 1} - 1) \lambda_r =  2^r - \lambda_r.
		\end{equation}
		Combining Eq.(\ref{Second_Layer_lambda}) and Eq.(\ref{Second_Layer_eq25}) we obtain
		\begin{equation}
		\sum_{i = 1}^{2^{r + 1} - 1} w_i \geq \frac{1}{\lambda_{r + 1} -
			\lambda_r}. \label{Second_Layer_eq26}
		\end{equation}
		By Eq.(\ref{Second_Layer_eq24}) and Eq.(\ref{Second_Layer_eq26}) we have
		\begin{equation}\label{Second_Layer_H_Ineq}
		H \leq (\lambda_{r + 1} - \lambda_r) \sum_{\mb{b} \in
			\mathrm{span}^c_{} (\mb{b}_1, \cdots, \mb{b}_{k - r - 1})}
		L_{\mb{b}}.
		\end{equation}
		Combining Eq.(\ref{Second_Layer_H_Ineq}) with Eq.(\ref{First_Layer_II}) and the definition of $H$ in Eq.(\ref{last_term_H}) we finally obtain,
		\fontsize{10pt}{13pt}\selectfont
		\begin{eqnarray} & &\mathrm{Err}^{(r + 1)} (\mb{b}_1, \cdots, \mb{b}_{k - r - 1}) \nonumber \\
		&\leq& \mathrm{OPT}_k + \lambda_{r + 1} \sum_{\mb{c} \in \mathrm{span}^c_{}
			(\mb{b}_1, \cdots, \mb{b}_{k - r - 1})} L_{\mb{c}} \nonumber \\ & &+
		\sum_{i = 1}^{k - r - 1} f_i (\mb{b}_1, \cdots, \mb{b}_{k - r
			- 1}) M_{\mb{b}_i}. 
        \end{eqnarray}
		\fontsize{11pt}{13pt}\selectfont
		
		This finishes the inductive step and completes the proof of the theorem.
		
	\end{proof}
	
	\section{Proof of Theorem \ref{Thm:Lower_Bound}}
	\begin{proof}
		The proof is constructive. For every $k$ and $\epsilon > 0$, we will explicitly give a matrix $\mb{A}$, when taken as input, the column-selection algorithm has error lower bounded by $\left(\frac{k}{2}+1+\frac{k}{2(2^k-1)}-\epsilon \right) \cdot \mathrm{OPT}_k$. For simplicity, we will construct a square matrix, i.e., $d=n$. Based on this construction, rectangular matrices such that $d|n$ work as well. We denote the matrix to be constructed as $\mb{A}$. The idea is to let this matrix be an approximate low rank matrix. That is, $\mb{A}$ is the product of two rank-$k$ matrices plus a sparse (noise) matrix.
		
		We assume the size of $\mb{A}$ satisfies $k|n$ and $(2^k-1)|n$. Let $p:=n/k$ and $q:=n/(2^k-1)$. $\mb{A}$ is constructed as follows.
		\begin{equation}
		\textbf{A} := \textbf{L} \textbf{R} + \textbf{I}_n,
		\end{equation}
		where $\textbf{I}_n$ is the identity matrix, $\textbf{L}$ is an $n \times k$ matrix defined as:
		\[
		\textbf{L} := \left(\begin{array}{cccc}
		\textbf{c}_1 & \textbf{c}_2 & \ldots & \textbf{c}_k
		\end{array}\right),
		\]
		where $\textbf{c}_i = \left( \underset{( i - 1) p}{\underbrace{0 \ldots 0}}
		\underset{p}{\underbrace{1 \ldots 1}}  \underset{( k - i) p}{\underbrace{0
				\ldots 0}} \right)^T$. For each $\textbf{c}_i$ all but $p$ elements are zero. $\textbf{R}$ is a $k \times n$ matrix defined as:
		\[\textbf{R} := \left(\begin{array}{cccc}
		\textbf{b}_1 \otimes \textbf{1}_q & \textbf{b}_2 \otimes \textbf{1}_q & \ldots & \textbf{b}_{2^k - 1} \otimes \textbf{1}_q
		\end{array}\right),
		\]
		where $\textbf{1}_q = ( 1 1 \ldots 1)^{T}$ is the all-one vector of size $q$; $\otimes$ is the Kronecker product; and $\textbf{b}_i$, which is a $k$-dimensional column vector, is the binary
		representation of $i$, e.g., $\textbf{b}_1 = ( 0 \ldots 0 0 1)^T$, $\textbf{b}_2 = ( 0 \ldots 0 1 0)^T$, $\textbf{b}_3 = ( 0 \ldots 0 1 1)^T$, etc. Below is a visualization of $\textbf{L}$ and $\textbf{R}$
		\[ \textbf{L} = \left(\begin{array}{cccc}
		1 & 0 & \ldots & 0\\
		\vdots & \vdots &  & \vdots\\
		1 & 0 & \ldots & 0\\
		0 & 1 & \ldots & 0\\
		\vdots & \vdots &  & \vdots\\
		0 & 1 & \ldots & 0\\
		\vdots & \vdots &  & \vdots\\
		0 & 0 & \ldots & 1\\
		\vdots & \vdots &  & \vdots\\
		0 & 0 & \ldots & 1
		\end{array}\right)_{n \times k} \]
        \[\textbf{R}=\left(\begin{array}{cccccccccc}
		0 &  & 0 & 0  & & 0 &  & 1 &  & 1\\
		\vdots & \ldots & \vdots & \vdots & \ldots & \vdots & \ldots & \vdots & \ldots & \vdots\\
		0 &  & 0 & 1 & & 1 &  & 1 & & 1\\
		1 &  & 1 & 0 & & 0 &  & 1 & & 1
		\end{array}\right)_{k \times n}\]
		
		The properties of $\mathbf{L}$ and $\mathbf{R}$ we will use are 1) The $1$s in different columns of $\mathbf{L}$ do not overlap. Any non-zero linear combination of the columns of $\mathbf{L}$ contains at least $q$ $1$s; 2) The columns of $\mathbf{R}$ contain all non-zero $k$-dimensional vectors (and repeat $q$ times).
		
		To prove the theorem we will show that no matter which $k$ columns of $\mathbf{A}$ are chosen to form the basis matrix, the induced approximation error is at least the desired lower bound for sufficiently large $n$.
		
		We will discuss two cases separately. In the first case we assume that the $k$ columns chosen from $\mathbf{A}$ has such a property: These $k$ columns $\mathbf{A}-\mathbf{I}_n$ are linear independent. In the second case we assume they are linear dependent.
		
		Now consider the first case. Let $(\mb{U}, \mb{V})$ be the output of the column-selection algorithm for input matrix $\mb{A}$. $\mb{U}$ consists of $k$ columns of $\mb{A}$. Let the indices of these $k$ columns be $i_1,\ldots,i_k$. Let $\mb{U}^-$ be the matrix consisting of the columns $i_1,\ldots,i_k$ of $\mb{A}-\mb{I}_n$. Let us consider the rank-$k$ approximation of $\mb{A}-\mb{I}_n$. $\mb{U}^-$ must be an optimal basis matrix because $\mb{A}-\mb{I}_n=\mb{L} \mb{R}$ is of rank $k$ and by our assumption $\mb{U}^-$ consists of $k$ linear independent columns, each is a linear combination of the columns of $\mb{L}$. So there is a $k$ by $n$ matrix $\mb{V}^-$ such that $\mb{A}-\mb{I}_n = \mb{U}^- \mb{V}^-$.
		
		Our local goal is to show $\mb{V} = \mb{V}^-$ if $n$ is sufficiently large. Consider the column-wise approximation error of $\mb{A}$ by $\mb{U} \mb{V}$. Let $\mb{A}=(\mb{a}_1,\ldots,\mb{a}_n)$, $\mb{V}=(\mb{v}_1,\ldots,\mb{v}_n)$, and $\mb{V}^-=(\mb{v}^-_1,\ldots,\mb{v}^-_n)$. For brevity, we will use $|\cdot|$ to represent $\|\|_F$. For column $i$, on the one hand we have,
		
		\begin{eqnarray}\label{Triangle_LB}
		 & | \textbf{U} \textbf{v}_i - \textbf{a}_i | \nonumber \\
        = & | \textbf{U} \textbf{v}_i -
		\textbf{U} \textbf{v}^-_i + \textbf{U} \textbf{v}^-_i - \textbf{U}^- \textbf{v}^-_i + \textbf{U}^- \textbf{v}^-_i - \textbf{a}_i | \nonumber\\
		 \geq  & | \textbf{U} (\textbf{v}_i - \textbf{v}^-_i) | - | (\textbf{U}  -
		\textbf{U}^-) \textbf{v}^-_i | - | \textbf{U}^- \textbf{v}^-_i - \textbf{a}_i |
		\end{eqnarray}
		On the other hand,
		\begin{eqnarray}\label{Triangle_UB}
		& | \textbf{U} \textbf{v}_i - \textbf{a}_i |  \nonumber \\
         = & | \textbf{U} \textbf{v}_i -
		\textbf{U} \textbf{v}^-_i + \textbf{U} \textbf{v}^-_i - \textbf{U}^- \textbf{v}^-_i + \textbf{U}^- \textbf{v}^-_i - \textbf{a}_i | \nonumber\\
		 \le & | \textbf{U} (\textbf{v}_i - \textbf{v}^-_i) | + | (\textbf{U}  -
		\textbf{U}^-) \textbf{v}^-_i | + | \textbf{U}^- \textbf{v}^-_i - \textbf{a}_i |
		\end{eqnarray}
		Let us analyze the three terms in the RHS of Eq.(\ref{Triangle_LB}) and Eq.(\ref{Triangle_UB}) respectively. If $\mb{V} \neq \mb{V}^-$, there must be a column $i$ such that $\mb{v}_i \neq \mb{v}^-_i$. So the first term $\textbf{U} (\textbf{v}_i - \textbf{v}^-_i)$ is a non-zero linear combination of the columns of $\mb{U}$. As $\mb{U}$ consists of $k$ columns of $\mb{L}\mb{R}+\mb{I}_n$, by the construction of $\mb{L}$ and $\mb{R}$, it is not difficult to see that when $\mb{v}_i \neq \mb{v}^-_i$, $|\textbf{U} (\textbf{v}_i - \textbf{v}^-_i)| \ge p-1$. For the second term $|(\textbf{U}  -  \textbf{U}^-) \textbf{v}^-_i| $, since each column of $\textbf{U}  -  \textbf{U}^-$ is of form $\mb{e}_j= \left( 0\dots010\dots0\right)^T$(only the $j$-th element is $1$) for some $j$, we have $|(\textbf{U}  -  \textbf{U}^-) \textbf{v}^-_i| = |\mb{v}^-_i| \le k$. For the third term $\textbf{U}^- \textbf{v}^-_i - \textbf{a}_i$, because $\textbf{U}^- \textbf{V}^- = \mb{L} \mb{R}$, we have $|\textbf{U}^- \textbf{v}^-_i - \textbf{a}_i| = |\mb{e}_i| = 1$.
		
		Combining the above argument, if $\mb{v}_i \neq \mb{v}^-_i$, then $| \textbf{U} \textbf{v}_i - \textbf{a}_i | \ge p-k-2$. If $\mb{v}_i = \mb{v}^-_i$, then $
		| \textbf{U} \textbf{v}_i - \textbf{a}_i | \le k+1$. Thus for $n$ sufficiently large so that $p=n/k$ is lager than $2k+3$, we must have $\mb{V} = \mb{V}^-$.
		
		Now we are able to calculate the column-wise approximation error. For $i$ such that $\mb{a}_i$ is a column of $\mb{U}$, the approximation error of this column is zero. For $i$ such that $\mb{a}_i$ is not a column of $\mb{U}$, we have:
		\begin{eqnarray*}
			 & & | \textbf{U} \textbf{v}_i - \textbf{a}_i | \\
             & = & | \textbf{U} \textbf{v}_i -
			\textbf{U} \textbf{v}^-_i + \textbf{U} \textbf{v}^-_i - \textbf{U}^- \textbf{v}^-_i + \textbf{U}^- \textbf{v}^-_i - \textbf{a}_i | \nonumber\\
			 & = & | (\textbf{U}  -
			\textbf{U}^-) \textbf{v}^-_i + \textbf{U}^- \textbf{v}^-_i - \textbf{a}_i |\nonumber\\
			& = & | (\textbf{U}  -
			\textbf{U}^-) \textbf{v}^-_i + \textbf{e}_i |.
		\end{eqnarray*}
		As argued above, each column of $\textbf{U}  -  \textbf{U}^-$ is of form $\mb{e}_j$ for some $j$. Note that $i$ must be different to $j$ since $\mb{a}_i$ is not a column of $\mb{U}$. Thus we have $| (\textbf{U}  -  \textbf{U}^-) \textbf{v}^-_i + \textbf{e}_i | = |\mb{v}^-_i|+1$.
		Therefore,
		\begin{align}
		 &| \textbf{U} \textbf{V} -\textbf{A} | \nonumber \\=& \sum_{i = 1}^n | \textbf{U} \textbf{v}_i - \textbf{a}_i | \nonumber \\
		=& \sum_{i = 1}^n ( 1 + |\textbf{v}^-_i|) - 2 k \nonumber \\=& n - 2 k + \sum_{i =
			1}^n | \textbf{v}^-_i |,
		\end{align}
		where the second equation holds because for a column $i$ so that $\mb{a}_i$ is in $\mb{U}$, we have $|\mb{v}^-_i|=|\mb{v}_i|=1$; and the total number of $i$ such that $\mb{a}_i$ belongs to $\mb{U}$ is $k$.
		
		Now let us examine $\mb{v}^-_i$. Because $\textbf{A} - \textbf{I}_n = \mb{L} \mb{R}$ contains all $2^k - 1$
		non-zero linear combinations of the column vectors of $\textbf{L}$; and recall that $\mb{U}^-$ consists of $k$ linear independent vectors, each is a linear combination of the vectors of $\mb{L}$, $\mb{V}^-$ must be a column-shuffled version of $\mb{R}$, i.e., the ordering of the columns changed. From this fact and the construction of $\mb{R}$ it is not difficult to see that \begin{equation}
		\sum_{i = 1}^n | \textbf{v}^-_i | = q k 2^{k - 1}.
		\end{equation}
		So the approximation error for the column-selection algorithm is $| \textbf{A} - \textbf{U} \textbf{V} | = n + q k 2^{k - 1} - 2k$. Let us compare this approximation error to that of the optimal approximation. Note that $(\mb{L},\mb{R})$ yields an approximation error $n$. So the optimal error is at most $n$. Thus we have the approximation ratio of the column-selection algorithm is at least $\frac{n + q k 2^{k - 1} - 2 k}{n} = 1 + \frac{k}{2} + \frac{k}{2 ( 2^k  - 1)} - \frac{2 k}{n}$. By letting $n$ larger than $2k/\epsilon$ we prove the lower bound.
		
		Finally consider the second case that the $k$ columns chosen from $\mathbf{A}$ has the property that the corresponding $k$ columns of $\mathbf{A}-\mathbf{I}_n$ are linear dependent. It is not difficult to check, from the construction of $\mb{L}$ and $\mb{R}$, that the columns of $\mb{L}$ which cannot be represented by the $k$ columns of $\mathbf{A}-\mathbf{I}_n$ must induce a total approximation error at least $pq-n$. So the approximation ratio is no less than $\frac{pq-n}{n}=\frac{n}{k(2^k-1)}-1$, which goes to infinity as $n$ getting large. This complete the proof.
	\end{proof}
	\section{Proof of Theorem \ref{Thm:Boolean}}
Let $\mc{S}_1,\ldots,\mc{S}_{2^k -1}$ be an ordering of the $2^k -1$ non-empty subsets of $[k]$ so that $n_{\mc{S}_1} \le \ldots \le n_{\mc{S}_{2^k-1}}$, and $\mc{S}_0=\emptyset$. As described in Introduction, we will construct $\mc{B}_1,\ldots,\mc{B}_k$ in such a way that 1) $\mc{B}_i$ is a Boolean combination of $\mc{D}_{\mc{S}_l}$ for all $\ell \in [2^k-1]$; and 2) for all $\ell \in [2^k-1]$
\begin{equation}\label{ap_Boolean_Decomposition}
	\mc{U}_{\mc{S}_\ell} \symdiff \left(\bigcup_{i \in S_\ell} \mc{B}_i \right)
	\subseteq  \left(\bigcup_{\substack{\ell' \geq \ell}} \left( \mc{U}_{\mc{S}_{\ell'}} \symdiff \mc{D}_{\mc{S}_{\ell'}} \right)\right).
\end{equation}

\begin{lemma}\label{ap_lemma:apprm_ratio_boolean}
	If Eq.(\ref{ap_Boolean_Decomposition}) is true for all $\ell \in [2^k -1]$, then the approximation ratio induced by the basis matrix whose columns are $\mc{B}_1,\ldots,\mc{B}_k$ is at most $2^k$.
\end{lemma}
\begin{proof}
	First we have for all $\ell$
	\begin{eqnarray*}
		\Abs{\mc{A}_j \symdiff \left(\bigcup_{i \in \mc{S}_\ell} \mc{B}_i\right)} & \leq & \abs{\mc{A}_j \symdiff \mc{U}_{S_\ell}} + \abs{\mc{U}_{\mc{S}_{\ell}}\symdiff (\bigcup_{i \in S_{\ell}}\mc{B}_i)} \nonumber \\
		& \leq & \abs{\mc{A}_j \symdiff \mc{U}_{\mc{S}_\ell}} + \sum_{\ell' \geq \ell} \abs{\mc{D}_{\mc{S}_{\ell'}} \symdiff \mc{U}_{\mc{S}_{\ell'}}}.
	\end{eqnarray*}
	Take summation both sides for all $j \in \mc{J}_{\mc{S}_{\ell}}$, we have

	\begin{align*}
	& \sum_{j \in \mc{J}_{\mc{S}_{\ell}}} \Abs{\mc{A}_j \symdiff \left(\bigcup_{i \in S_\ell} \mc{B}_i\right)} \\ \leq & ( \sum_{j \in \mc{J}_{\mc{S}_{\ell}}} \abs{\mc{A}_j \symdiff \mc{U}_{\mc{S}_\ell}} ) + \sum_{\ell' \geq \ell} n_{\mc{S}_{\ell}} \abs{\mc{D}_{\mc{S}_{\ell'}} \symdiff \mc{U}_{\mc{S}_{\ell'}}}.
	\end{align*}

	Let $\mathrm{Err}(\mc{B}_1,\ldots,\mc{B}_k)$ denote the approximation error induced by the basis $\mc{B}_1,\ldots,\mc{B}_k$; and $\OPT_k$ be the error of the optimal solution. Taking summation both sides for all $\ell \in [2^k-1]$ in the above inequality and then add $( \sum_{j \in \mc{J}_{\mc{S}_{0}}} \abs{\mc{A}_j} )$ both sides, we have
	\[
	\begin{split}
	\mathrm{Err}(\mc{B}_1,\ldots,\mc{B}_k) \leq \OPT_k + \sum_{\ell = 1}^{2^k-1} \sum_{\ell' \geq \ell} n_{\mc{S}_{\ell}} \abs{\mc{D}_{\mc{S}_{\ell'}} \symdiff \mc{U}_{\mc{S}_{\ell'}}}
	\end{split}
	\]
	Since $n_{\mc{S}_{\ell}} \leq n_{\mc{S}_{\ell'}}$ for all $\ell' \geq \ell$, the last inequality becomes
	\[
	\begin{split}
	\mathrm{Err}(\mc{B}_1,\dots,\mc{B}_k) \leq \OPT_k + \sum_{\ell = 1}^{2^k-1} \sum_{\ell' \geq \ell} n_{\mc{S}_{\ell'}} \abs{\mc{D}_{\mc{S}_{\ell'}} \symdiff \mc{U}_{\mc{S}_{\ell'}}}
	\end{split}
	\]
	Change the order of summations for $\ell$ and $\ell'$, and use the fact that
	$n_{\mc{S}_{\ell'}} \abs{\mc{D}_{\mc{S}_{\ell'}} \symdiff \mc{U}_{\mc{S}_{\ell'}}}  \leq \sum_{j \in \mc{J}_{\mc{S}_{\ell'}}} \abs{\mc{A}_j \symdiff \mc{U}_{\mc{S}_{\ell'}}}$, we finally get
	\[
	\begin{split}
	 & \mathrm{Err}(\mc{B}_1,\ldots,\mc{B}_k)  \\ 
    \leq & \OPT_k + \sum_{\ell' = 1}^{2^k -1} \ell' \sum_{j \in \mc{J}_{\mc{S}_{\ell'}}} \abs{\mc{A}_j \symdiff \mc{U}_{\mc{S}_{\ell'}}}\\
	\leq & \OPT_k + (2^k-1) \sum_{\ell' = 1}^{2^k -1} \sum_{j \in \mc{J}_{\mc{S}_{\ell'}}} \abs{\mc{A}_j \symdiff \mc{U}_{\mc{S}_{\ell'}}}\\
	\leq & 2^k\OPT_k.
	\end{split}
	\]
	This completes the proof of the lemma.
\end{proof}

Before describing the construction, we state the following results which will be frequently used in the rest of this section. Let $\mc{A,B,C,D}$ be sets.
\begin{lemma}
	\begin{enumerate}
		\item $(\mc{A} \cup \mc{B}) \setminus \mc{C} = (\mc{A} \setminus \mc{C}) \cup (\mc{B} \setminus \mc{C})$
		\item $(\mc{A} \cap \mc{B}) \setminus \mc{C} = (\mc{A} \setminus \mc{C}) \cap (\mc{B} \setminus \mc{C})$
		\item $\mc{A} \setminus (\mc{B} \cup \mc{C}) = (\mc{A} \setminus \mc{B}) \cap (\mc{A} \setminus \mc{C})$
		\item $\mc{A} \setminus (\mc{B} \cap \mc{C}) = (\mc{A} \setminus \mc{B}) \cup (\mc{A} \setminus \mc{C})$
	\end{enumerate}
\end{lemma}

\begin{corollary}
	\label{ap_COR:BoundingSetDiff}
	By the previous lemma:
	\begin{enumerate}
		\item $(\mc{A} \cup \mc{B}) \setminus (\mc{C} \cup \mc{D}) \subseteq (\mc{A}\setminus \mc{C}) \cup (\mc{B} \setminus \mc{D})$
		\item $(\mc{A} \cap \mc{B}) \setminus (\mc{C} \cap \mc{D}) \subseteq (\mc{A}\setminus \mc{C}) \cup (\mc{B} \setminus \mc{D})$
	\end{enumerate}
	In particular:
	\begin{enumerate}
		\item $(\mc{A} \cup \mc{B}) \symdiff (\mc{C} \cup \mc{D}) \subseteq (\mc{A}\symdiff \mc{C}) \cup (\mc{B} \symdiff \mc{D})$
		\item $(\mc{A} \cap \mc{B}) \symdiff (\mc{C} \cap \mc{D}) \subseteq (\mc{A}\symdiff \mc{C}) \cup (\mc{B} \symdiff \mc{D})$
	\end{enumerate}
\end{corollary}

\begin{lemma}
	Triangle inequalities:
	\begin{enumerate}
		\item $\mc{A} \setminus \mc{B} \subseteq (\mc{A} \setminus \mc{C} ) \cup (\mc{C} \setminus \mc{B})$.
		\item $\mc{A} \symdiff \mc{B} \subseteq (\mc{A} \symdiff \mc{C}) \cup (\mc{C} \symdiff \mc{B})$.
	\end{enumerate}
\end{lemma}

Now we begin to construct $\mc{B}_1,\ldots,\mc{B}_{2^k-1}$ so that Eq.(\ref{ap_Boolean_Decomposition}) can be satisfied. The first step is to define sets $\mc{E}_i^{\ell}$. For $i \in [k]$ and $\ell \in [2^k-1]$, define
\[
\mc{E}^\ell_i := \bigcap_{\substack{\ell' \geq \ell :\\ i \in \mc{S}_{\ell'}}} \mc{D}_{\mc{S}_{\ell'}}
\]
Clearly, by definition we have $\mc{E}^{\ell}_i \subseteq \mc{E}^{\ell'}_i$ for every
$1 \leq \ell \leq \ell' \leq 2^k-1$.


We construct $\mc{B}_1,\ldots,\mc{B}_{2^k-1}$ based on the sets $\mc{D}_{\mc{S}_1},\dots,\mc{D}_{\mc{S}_{2^k-1}}$. The key idea of the construction
is to obtain that for all $\ell$
\begin{equation}
	\label{ap_EQ:OverviewOfConstruction}
	\bigcup_{i \in \mc{S}_\ell} \mc{B}_i = \left(\bigcup_{i \in \mc{S}_\ell} \mc{E}^\ell_i \right) \cup \mc{R}_\ell \enspace ,
\end{equation}
in such a way that it is to possible express $R_\ell$ as a subset of the union of sets
$\mc{U}_{\mc{S}_{\ell'}} \symdiff \mc{D}_{\mc{S}_{\ell'}}$ for $\ell' \geq \ell$. 
The following lemma states that the construction given in Eq.(\ref{ap_EQ:OverviewOfConstruction}) implies Eq.(\ref{ap_Boolean_Decomposition}).
\begin{lemma}\label{ap_Lemma:Decomposition} Assuming Eq.(\ref{ap_EQ:OverviewOfConstruction}) is true, then
	\begin{equation}
		\mc{U}_{\mc{S}_\ell} \symdiff \left(\bigcup_{i \in \mc{S}_\ell} \mc{B}_i \right)
		\subseteq \left(\bigcup_{\substack{\ell' \geq \ell:\\ \mc{S}_\ell \cap \mc{S}_{\ell'} \neq \emptyset}} \left( \mc{U}_{\mc{S}_{\ell'}} \symdiff \mc{D}_{\mc{S}_{\ell'}} \right)\right) \cup \mc{R}_\ell.
	\end{equation}
\end{lemma}

\begin{proof}
	Using Eq.(\ref{ap_EQ:OverviewOfConstruction}) we have
	\[
	\begin{split}
	& \mc{U}_{\mc{S}_\ell} \symdiff \left(\bigcup_{i \in \mc{S}_\ell} \mc{B}_i \right)  \\
    \subseteq &
	\left(\mc{U}_{\mc{S}_\ell} \setminus \left[\bigcup_{i \in \mc{S}_\ell} \mc{E}^\ell_i \right]\right)  \cup  \left(\left[\bigcup_{i \in \mc{S}_\ell} \mc{E}^\ell_i \right] \setminus \mc{U}_{\mc{S}_\ell}\right) \cup \mc{R}_\ell\\
	\subseteq &\left(\bigcup_{i \in \mc{S}_\ell} \left(\mc{U}_i \setminus \mc{E}^\ell_i \right)\right) \cup (\mc{D}_{\mc{S}_\ell} \setminus \mc{U}_{\mc{S}_\ell}) \cup \mc{R}_\ell
	\end{split}
	\]
	where we used Corollary~\ref{ap_COR:BoundingSetDiff} and that $\bigcup_{i \in \mc{S}_\ell} \mc{E}^\ell_i \subseteq \mc{D}_{\mc{S}_\ell}$. Next we claim that for $i \in S_{\ell}$,
	\[
	\mc{U}_i \setminus \mc{E}^\ell_i \subseteq \bigcup_{\substack{\ell'\geq\ell:\\i \in \mc{S}_{\ell'}}} \left(\mc{U}_{\mc{S}_{\ell'}} \setminus \mc{D}_{\mc{S}_{\ell'}}\right).
	\]
	To see this, first note that
	\[
	\mc{U}_i \subseteq \bigcap_{\substack{\ell' \geq \ell
			:\\ i \in \mc{U}_{\ell'}}} \mc{U}_{\mc{S}_{\ell'}}.
	\]
	Thus
	\[
	\mc{U}_i \setminus \mc{E}^\ell_i \subseteq \left(\bigcap_{\substack{\ell' \geq \ell
			:\\ i \in \mc{S}_{\ell'}}} \mc{U}_{\mc{S}_{\ell'}}\right) \setminus \left(\bigcap_{\substack{\ell' \geq \ell
			:\\ i \in \mc{S}_{\ell'}}} \mc{D}_{\mc{S}_{\ell'}}\right) \subseteq \bigcup_{\substack{\ell'\geq\ell:\\i \in \mc{S}_{\ell'}}} \left(\mc{U}_{\mc{S}_{\ell'}} \setminus \mc{D}_{\mc{S}_{\ell'}}\right).
	\]
	Therefore
	\begin{equation}\label{ap_u_diff_b}
		\begin{split}
			& \mc{U}_{\mc{S}_\ell} \symdiff \left(\bigcup_{i \in \mc{S}_\ell} \mc{B}_i \right) \\ 
			 \subseteq & \left(\bigcup_{i \in \mc{S}_\ell} \bigcup_{\substack{\ell'\geq\ell:\\i \in \mc{S}_{\ell'}}} \left(\mc{U}_{\mc{S}_{\ell'}} \setminus \mc{D}_{\mc{S}_{\ell'}}\right)\right) \cup (\mc{D}_{\mc{S}_\ell} \setminus \mc{U}_{\mc{S}_\ell})  \cup \mc{R}_\ell \\  \subseteq & \left(\bigcup_{\substack{\ell' \geq \ell:\\ S_\ell \cap \mc{S}_{\ell'} \neq \emptyset}} \left( \mc{U}_{\mc{S}_{\ell'}} \symdiff \mc{D}_{\mc{S}_{\ell'}} \right)\right) \cup \mc{R}_\ell.
		\end{split}
	\end{equation}
\end{proof}

To satisfy Eq.(\ref{ap_EQ:OverviewOfConstruction}) we proceed as
follows: We initially let $\mc{B}_i = \mc{E}^1_i$ for all $i \in [k]$. Next,
for each $\ell \in [2^k-1]$, we will simply be adding
\[
\left(\bigcup_{i \in \mc{S}_\ell} \mc{E}^\ell_i \right) \setminus \left(\bigcup_{i \in \mc{S}_\ell} \mc{E}^1_i \right)
\]
to $\bigcup_{i \in S_\ell} \mc{B}_i$, which means that
Eq.\eqref{ap_EQ:OverviewOfConstruction} is satisfied. But now, in
order to be able to bound the set $\mc{R}_{\ell'}$ for all $\ell'$
simultaneously, we will do this carefully piece by piece using the
ordering of the sets $\mc{S}_1,\dots,\mc{S}_{2^k-1}$.

The main step of the construction is this: For $\ell_1<\ell_2$ such that $i \in \mc{S}_{\ell_1} \cap \mc{S}_{\ell_2}$, define the set $\mc{F}_i^{\ell_1,\ell_2}$ by
\[
\mc{F}_i^{\ell_1,\ell_2} := \mc{E}_i^{\ell_1+1} \setminus \left[\bigcup_{i' \in \mc{S}_{\ell_2}} \mc{E}^{\ell_1}_{i'}\right].
\]
We now define the rank-$k$ solution $(\mc{B}_1,\dots,\mc{B}_k)$ by
\begin{equation}\label{ap_Def:B}
	\mc{B}_i := \mc{E}_i^1 \cup \left(\bigcup_{\substack{\ell_1 < \ell_2:\\i \in \mc{S}_{\ell_1} \cap \mc{S}_{\ell_2}}} \mc{F}_i^{\ell_1,\ell_2} \right)\enspace .
\end{equation}
Let now $\ell$ be fixed. Then
\[
\bigcup_{i \in \mc{S}_\ell} \mc{B}_i = \left(\bigcup_{i \in \mc{S}_\ell} \mc{E}_i^1\right) \cup \left(\bigcup_{i \in S_\ell}\bigcup\limits_{\substack{\ell_1 < \ell_2:\\i \in \mc{S}_{\ell_1} \cap \mc{S}_{\ell_2}}} \mc{F}_i^{\ell_1,\ell_2} \right).
\]

The following lemma gives the formula of $\mc{R}_{\ell}$ in Eq.(\ref{ap_EQ:OverviewOfConstruction})
\begin{lemma}\label{ap_Def:R}
	\[
	\bigcup_{i \in \mc{S}_\ell} \mc{B}_i = \left(\bigcup_{i \in \mc{S}_\ell} \mc{E}_i^\ell \right) \cup \left(\bigcup_{i \in S_\ell}\bigcup_{\substack{\ell \leq \ell_1 < \ell_2:\\i \in \mc{S}_{\ell_1} \cap \mc{S}_{\ell_2}}} \mc{F}_i^{\ell_1,\ell_2} \right).
	\]
	Thus we can choose
	\[
	\mc{R}_\ell = \bigcup_{i \in \mc{S}_\ell}\bigcup_{\substack{\ell \leq \ell_1 < \ell_2:\\i \in \mc{S}_{\ell_1} \cap \mc{S}_{\ell_2}}} \mc{F}_i^{\ell_1,\ell_2} \enspace .
	\]
\end{lemma}

\begin{proof} When $\ell_1 < \ell$, $\mc{F}_i^{\ell_1,\ell_2} \subseteq \mc{E}_i^{\ell_1 +1} \subseteq \mc{E}_i^\ell$, therefore,
	\[\bigcup_{\substack{\ell_1 < \ell \\ \ell_1 < \ell_2:\\i \in \mc{S}_{\ell_1} \cap \mc{S}_{\ell_2}}} \mc{F}_i^{\ell_1,\ell_2} \subseteq \mc{E}_i^\ell\]
	And then,
	\[ \bigcup_{i \in S_\ell}\bigcup_{\substack{\ell_1 < \ell_2:\\i \in \mc{S}_{\ell_1} \cap \mc{S}_{\ell_2}}} \mc{F}_i^{\ell_1,\ell_2} \subseteq \left(\bigcup_{i \in \mc{S}_\ell} \mc{E}_i^\ell \right) \cup \left(\bigcup_{i \in S_\ell}\bigcup_{\substack{\ell \leq \ell_1 < \ell_2:\\i \in \mc{S}_{\ell_1} \cap \mc{S}_{\ell_2}}} \mc{F}_i^{\ell_1,\ell_2} \right) \]
	On the other hand,
	\[
	\bigcup_{i \in S_\ell}\bigcup_{\substack{\ell \leq \ell_1 < \ell_2:\\i \in \mc{S}_{\ell_1} \cap \mc{S}_{\ell_2}}} \mc{F}_i^{\ell_1,\ell_2} \subseteq \bigcup_{i \in S_\ell}\bigcup_{\substack{\ell_1 < \ell_2:\\i \in \mc{S}_{\ell_1} \cap \mc{S}_{\ell_2}}} \mc{F}_i^{\ell_1,\ell_2}
	\]
	Hence,
	\[
	\begin{split}
	& \left(\bigcup_{i \in \mc{S}_\ell} \mc{E}_i^\ell \right) \cup \left(\bigcup_{i \in S_\ell}\bigcup_{\substack{\ell \leq \ell_1 < \ell_2:\\i \in \mc{S}_{\ell_1} \cap \mc{S}_{\ell_2}}} \mc{F}_i^{\ell_1,\ell_2} \right) \\
	= & \left(\bigcup_{i \in \mc{S}_\ell} \mc{E}_i^\ell \right) \cup \left(\bigcup_{i \in S_\ell}\bigcup_{\substack{\ell_1 < \ell_2:\\i \in \mc{S}_{\ell_1} \cap \mc{S}_{\ell_2}}} \mc{F}_i^{\ell_1,\ell_2} \right) \\
	= & \bigcup_{i \in \mc{S}_\ell} \left[ \mc{E}_i^\ell \cup \left( \bigcup_{\substack{\ell_1 < \ell_2:\\i \in \mc{S}_{\ell_1} \cap \mc{S}_{\ell_2}}} \mc{F}_i^{\ell_1,\ell_2} \right) \right] \\
	\end{split}\]
	Since
	\[\mc{B}_i = \mc{E}_i^1 \cup \left(\bigcup_{\substack{\ell_1 < \ell_2:\\i \in \mc{S}_{\ell_1} \cap \mc{S}_{\ell_2}}} \mc{F}_i^{\ell_1,\ell_2} \right)\subseteq \mc{E}_i^\ell \cup \left(\bigcup_{\substack{\ell_1 < \ell_2:\\i \in \mc{S}_{\ell_1} \cap \mc{S}_{\ell_2}}} \mc{F}_i^{\ell_1,\ell_2} \right)\]
	we have
	\[
	\bigcup_{i \in \mc{S}_\ell} \mc{B}_i \subseteq \bigcup_{i \in \mc{S}_\ell} \left[ \mc{E}_i^\ell \cup \left(\bigcup_{\substack{\ell_1 < \ell_2:\\i \in \mc{S}_{\ell_1} \cap \mc{S}_{\ell_2}}} \mc{F}_i^{\ell_1,\ell_2} \right) \right]
	\]
	Now it suffices to prove
	\[
	\bigcup_{i \in \mc{S}_\ell} \left[ \mc{E}_i^\ell \cup \left(\bigcup_{\substack{\ell_1 < \ell_2:\\i \in \mc{S}_{\ell_1} \cap \mc{S}_{\ell_2}}} \mc{F}_i^{\ell_1,\ell_2} \right) \right] \subseteq \bigcup_{i \in \mc{S}_\ell} \mc{B}_i
	\]
	Note that by \eqref{ap_Def:B},
	\[
	\bigcup_{\substack{\ell_1 < \ell_2:\\i \in \mc{S}_{\ell_1} \cap \mc{S}_{\ell_2}}} \mc{F}_i^{\ell_1,\ell_2} \subseteq \mc{B}_i
	\]
	Thus we only need to show that for all $\ell$
	\begin{equation}\label{ap_R_ell}
		\bigcup_{i \in \mc{S}_\ell} \mc{E}^\ell_i \subseteq \bigcup_{i \in \mc{S}_\ell} \mc{B}_i.
	\end{equation}
	We prove Eq.(\ref{ap_R_ell}) by induction that for all $\ell'\leq \ell$ we have
	\begin{equation}
		\bigcup_{i \in \mc{S}_\ell} \mc{E}^{\ell'}_i \subseteq \bigcup_{i \in \mc{S}_\ell} \mc{B}_i
	\end{equation}
	The base case $\ell'=1$ holds since $\mc{E}^1_i \subseteq \mc{B}_i$ by
	definition of $\mc{B}_i$, for all $i$. Let now $\ell'<\ell$ and $i \in \mc{S}_\ell$.
	If $i \notin \mc{S}_{\ell'}$ we have $\mc{E}^{\ell'+1}_i=\mc{E}^{\ell'}_i$, and hence
	$\mc{E}^{\ell'+1}_i \subseteq \bigcup_{i' \in \mc{S}_\ell} \mc{B}_{i'}$ by induction.
	Suppose now $i \in \mc{S}_{\ell'}$. By induction we have
	\[
	\bigcup_{i' \in \mc{S}_{\ell}} \mc{E}^{\ell'}_{i'} \subseteq \bigcup_{i' \in \mc{S}_\ell} \mc{B}_{i'} \enspace ,
	\]
	Since $i \in \mc{S}_{\ell'} \cap \mc{S}_{\ell}$ we have
	\[
	\mc{F}^{\ell',\ell}_i = \mc{E}^{\ell'+1}_i \setminus \left[\bigcup_{i' \in \mc{S}_{\ell}} \mc{E}^{\ell'}_{i'}\right] \subseteq \mc{B}_i \enspace ,
	\]
	and we can conclude that $\mc{E}^{\ell'+1}_i \subseteq \bigcup_{i' \in
		\mc{S}_\ell} \mc{B}_{i'}$ in this case as well. Since this holds for any $i
	\in \mc{S}_\ell$, this concludes the proof.

\end{proof}

Finally, that $\mc{R}_{\ell}$ is a subset of the union of of sets $\mc{U}_{\mc{S}_{\ell'}} \symdiff \mc{D}_{\mc{S}_{\ell'}}$ for $\ell' \geq \ell$ can be obtained from the following lemma.
\begin{lemma}\label{ap_UB_F}
	Let $\ell_1 < \ell_2$ such that $i \in \mc{S}_{\ell_1} \cap \mc{S}_{\ell_2}$. Then
	\[
	\mc{F}^{\ell_1,\ell_2}_i
	\subseteq \left( \mc{D}_{\mc{S}_{\ell_2}} \setminus \mc{U}_{\mc{S}_{\ell_2}}\right) \cup \left(\bigcup_{i' \in \mc{S}_{\ell_2}} \bigcup_{\substack{\ell'\geq \ell_1:\\ i' \in \mc{S}_{\ell'}}} \left(\mc{U}_{\mc{S}_{\ell'}} \setminus \mc{D}_{\mc{S}_{\ell'}}\right)\right).
	\]
\end{lemma}
\begin{proof}
	We have
	\[
	\begin{split}
	& \mc{F}^{\ell_1,\ell_2}_i  \\    
     \subseteq & \left(\mc{E}_i^{\ell_1+1} \setminus \mc{U}_{\mc{S}_{\ell_2}}\right) \cup \left(\mc{U}_{\mc{S}_{\ell_2}} \setminus \left[\bigcup_{i' \in \mc{S}_{\ell_2}} \mc{E}^{\ell_1}_{i'}\right]\right) \\
	 \subseteq & \left( \mc{D}_{\mc{S}_{\ell_2}} \setminus \mc{U}_{\mc{S}_{\ell_2}}\right) \cup \left(\left[\bigcup_{i' \in \mc{S}_{\ell_2}} \bigcap_{\substack{\ell'\geq \ell_1:\\ i' \in \mc{S}_{\ell'}}} \mc{U}_{\mc{S}_{\ell'}}\right]\setminus\left[\bigcup_{i' \in \mc{S}_{\ell_2}} \bigcap_{\substack{\ell'\geq \ell_1:\\ i' \in \mc{S}_{\ell'}}} \mc{D}_{\mc{S}_{\ell'}}\right]\right)\\  \subseteq & \left( \mc{D}_{\mc{S}_{\ell_2}} \setminus \mc{U}_{\mc{S}_{\ell_2}}\right) \cup \left(\bigcup_{i' \in \mc{S}_{\ell_2}} \bigcup_{\substack{\ell'\geq \ell_1:\\ i' \in \mc{S}_{\ell'}}} \left(\mc{U}_{\mc{S}_{\ell'}} \setminus \mc{D}_{\mc{S}_{\ell'}}\right)\right),
	\end{split}
	\]
	where we used that $\mc{E}^{\ell_1}_i \subseteq \mc{D}_{\mc{S}_{\ell_2}}$ and $\mc{U}_{i'} \subseteq \mc{U}_{\mc{S}_{\ell'}}$ for every $i' \in \mc{S}_{\ell_2}$ and $\ell'\geq \ell_1$ for which $i' \in \mc{S}_{\ell'}$.
\end{proof}

Combining \eqref{ap_u_diff_b} , Lemma \ref{ap_Def:R} and Lemma \ref{ap_UB_F}, we conclude that
\[
\begin{split}
\mc{U}_{S_\ell} \symdiff \left(\bigcup_{i \in \mc{S}_\ell} \mc{B}_i \right)
\subseteq  \left(\bigcup_{\substack{\ell' \geq \ell}} \left( \mc{U}_{\mc{S}_{\ell'}} \symdiff \mc{D}_{\mc{S}_{\ell'}} \right)\right)
\end{split}
\]
for all $\ell$.
	
	\section{Proof of Theorem \ref{Thm:NP_Hardness}}
	The lemma below is an adaptation of Lemma 4.2 in \cite{TIT:RothV08} from the $\pmo$ case to the $\zo$ case.
	\begin{lemma}\label{hardness_lemma_1}
		Let $\mb{W}$ be an $n\times n$ matrix and let $m \geq 1$, and define $\mb{W}'=\mb{W}
		\tensor \mb{J}_m$, where $\mb{J}_m := \mb{J}_{m,m}$. Then
		\[
		\max_{\mb{u,v}} \mb{u}^\transpose \mb{W}'\mb{v} = m^2 \cdot \max_{\mb{x,y}}
		\mb{x}^\transpose \mb{Wy} \enspace ,
		\]
		where $\mb{u,v} \in \zo^{mn}$ and $\mb{x,y} \in \zo^n$, respectively.
		Furthermore, if $\mb{x}$ and $\mb{y}$ maximize $\mb{x}^\transpose \mb{Wy}$, then
		$\mb{u} = \mb{x} \tensor \mb{1}_m$ and $\mb{v} = \mb{y} \tensor \mb{1}_m$ maximize $\mb{u}^\transpose
		\mb{W'v}$.
		\label{LEM:BlockmatrixOptima}
	\end{lemma}
	
	\begin{proof}
		Consider first $\mb{u} = \mb{x} \tensor \mb{1}_d$ and $\mb{v} = \mb{y} \tensor \mb{1}_m$. Then
		\[
		\begin{split}
		& \mb{u}^\transpose(\mb{W} \tensor \mb{J}_m)\mb{v} \\
        &= (\mb{x} \tensor \mb{1}_m)^\transpose(\mb{W} \tensor \mb{J}_m)(\mb{y} \tensor \mb{1}_m) \\
		&=(\mb{x}^\transpose \mb{Wy})\tensor(\mb{1}_m^\transpose \mb{J}_m\mb{1}_m) = m^2 \cdot (\mb{x}^\transpose \mb{Wy}) \enspace .
		\end{split}
		\]
		Next, take $\mb{u}$ and $\mb{v}$ maximizing $\mb{u}^\transpose \mb{W'v}$. We show that $\mb{u}$
		and $\mb{v}$ can be brought to the form $\mb{u} = \mb{x} \tensor \mb{1}_m$ and $\mb{v} = \mb{y}
		\tensor \mb{1}_m$ without decreasing the value of $\mb{u}^\transpose \mb{W'v}$. We
		first fix $\mb{v}$ and bring $\mb{u}$ to the desired form, and then similarly
		bring $\mb{v}$ to the desired form.
		
		So fix $\mb{v}$, and let $\mb{z}=\mb{W'v}$. Note that $\mb{u}$ maximizing $\mb{u}^\transpose \mb{z}$
		must satisfy $\mb{u}_i=1$ when $\mb{z}_i>0$ and $\mb{u}_i=0$ when $\mb{z}_i<0$. Since
		$\mb{W}'=\mb{W} \tensor \mb{J}_m$ we have that $\mb{z}_{jm+1}=\mb{z}_{jm+2}= \dots =
		\mb{z}_{(j+1)m}$ for all $j=0,1,\dots,n-1$. Hence we can choose a
		maximizing $\mb{u}$ satisfying $\mb{u}_{jm+1}=\mb{u}_{jm+2}=\dots=\mb{u}_{(j+1)m}$ for all
		$j=0,1,\dots,n-1$ as well, meaning $\mb{u}=\mb{x} \tensor \mb{1}_m$ for suitable $\mb{x}
		\in \zo^n$. We can now fix $\mb{u}$ and in a similar way bring $\mb{v}$ to the
		form $\mb{v}=\mb{y} \tensor \mb{1}_m$ for suitable $\mb{y} \in \zo^n$.
	\end{proof}

	The following lemma, which is the $\zo$ analogue of Lemma 4.3 in
	\cite{TIT:RothV08}, is a direct consequence of Lindsey's
	Lemma. We state the proof for completeness.
	\begin{lemma}\label{hardness_lemma_2}
		Let $\mb{H}$ be a $m \times m$ Hadamard matrix. For every $\mb{x,y} \in \zo^m$,
		\[
		\abs{\mb{x}^\transpose \mb{Hy}} \leq m^{3/2} \enspace .
		\]
		\label{LEM:LindseyCor}
	\end{lemma}
	
	\begin{proof}
		First note
		\[
		\norm{\mb{Hy}}^2 = \mb{y}^\transpose (\mb{H}^\transpose \mb{H})\mb{y} = \mb{y}^\transpose (m\mb{I}) \mb{y} = m\cdot\norm{\mb{y}}^2 \enspace .
		\]
		We can then complete the proof by the Cauchy-Schwartz inequality,
		\[
		\abs{\mb{x}^\transpose \mb{Hy}} \leq \norm{\mb{x}^\transpose}\cdot\norm{\mb{Hy}} = \sqrt{m}\cdot\norm{\mb{x}}\cdot\norm{\mb{y}} \leq m^{3/2} \enspace .
		\]
	\end{proof}
	
	\begin{lemma}\label{hardness_lemma_3}
		Let $\mb{W}=(w_{ij})$ be a $n \times n$ $\zpmo$-matrix and let $\mb{H}$ be a
		$m\times m$ Hadamard matrix. Define the $(mn) \times (mn)$
		$\pmo$-block matrix $\widetilde{\mb{W}}=(\widetilde{\mb{W}}_{ij})$, where block
		$\widetilde{\mb{W}}_{ij}$ is given by
		\[
		\widetilde{\mb{W}}_{ij} = \begin{cases}
		w_{ij} \mb{J}_m & \text{ if } w_{ij}\neq 0\\
		\mb{H} & \text{ if } w_{ij} = 0
		\end{cases} \enspace .
		\]
		Let $\mb{W}' = \mb{W} \tensor \mb{J}_m$. Then for all $\mb{u,v} \in \zo^{mn}$,
		\[
		\Abs{\mb{u}^\transpose\widetilde{\mb{W}}\mb{v} - \mb{u}^\transpose \mb{W'v}} \leq n^2 \cdot m^{3/2} \enspace .
		\]
		\label{LEM:HadamardReplacement}
	\end{lemma}
	
	\begin{proof}
		This is by simple estimation.
		\[
		\begin{split}
		\Abs{\mb{u}^\transpose\widetilde{\mb{W}}\mb{v}-\mb{u}^\transpose \mb{W'v} }
		& = \Abs{\mb{u}^\transpose(\widetilde{\mb{W}}-\mb{W}')\mb{v}}\\
		& \leq n^2 \cdot \max_{\mb{x,y} \in \zo^m} \Abs{\mb{x}^\transpose \mb{H y}} \\
		& \leq  n^2 \cdot m^{3/2} \enspace ,
		\end{split}
		\]
		where the last inequality follows from Lemma~\ref{LEM:LindseyCor}.
	\end{proof}
	
	\begin{proof}\textbf{of Theorem \ref{Thm:NP_Hardness}}
		Suppose now that $\mb{W}$ is an $n \times n$ $\zpmo$-matrix. Let $m=2^\ell$
		be the smallest power of 2 that is greater than $4n^4$, and let $\mb{H}$
		be the $m \times m$ Sylvester Hadamard matrix. We then define
		$\widetilde{\mb{W}}$ and $\mb{W}'$ as in
		Lemma~\ref{LEM:HadamardReplacement}. Then
		\[
		\begin{split}
		&\Abs{\max_{\mb{u,v} \in \zo^{mn}} \mb{u}^\transpose \widetilde{\mb{W}}\mb{v}-m^2 \cdot \max_{\mb{x,y} \in \zo^n} \mb{x}^\transpose \mb{W y}}
		\\= &\Abs{\max_{\mb{u,v} \in \zo^{mn}} \mb{u}^\transpose \widetilde{\mb{W}} \mb{v}- \max_{\mb{u,v} \in \zo^{mn}} \mb{u}^\transpose \mb{W'v}}\\\leq & n^2 \cdot m^{3/2} \leq \frac{m^{1/2}}{2} \cdot m^{3/2} = \frac{m^2}{2} \enspace ,
		\end{split}
		\]
		where the first equality is by Lemma~\ref{LEM:BlockmatrixOptima} and
		the first inequality is by Lemma~\ref{LEM:HadamardReplacement}. \\
		Since
		the expression $m^2 \cdot \max_{\mb{x,y} \in \zo^n} \mb{x}^\transpose \mb{W y}$ is an integer
		multiple of $m^2$, the value \\ $\max_{\mb{u,v} \in \zo^{mn}} \mb{u}^\transpose
		\widetilde{\mb{W}}\mb{v}$ uniquely determines the value $\max_{\mb{x,y} \in \zo^n} \mb{x}^\transpose
		\mb{W y}$. This then gives the desired reduction.
	\end{proof}

\bibliography{example_paper}
\end{document}